\documentclass[12pt]{article}
\usepackage{amsmath}
\usepackage{amssymb}
\usepackage{mathrsfs}
\usepackage{amsthm}
\usepackage{authblk}
\usepackage{color}
\usepackage{graphicx}
\usepackage{caption}
\usepackage{subcaption}
\usepackage{soul}
\usepackage{bbm} 
\usepackage{ulem}
\usepackage{cancel} 
\usepackage{algorithm}
\usepackage{algorithmicx}
\usepackage{amssymb}
\theoremstyle{definition}
\newtheorem{definition}{Definition}
\newtheorem{theorem}{Theorem}
\newtheorem{lemma}{Lemma}
\newtheorem{corollary}{Corollary}
\newtheorem{remark}{Remark}
\newtheorem{proposition}{Proposition}
\usepackage{hyperref}
\usepackage[margin=1.2in]{geometry}
\allowdisplaybreaks[4]

\newcommand{\be}{\begin{equation}}
	\newcommand{\ee}{\end{equation}}
\newcommand{\bea}{\begin{eqnarray*}}
	\newcommand{\eea}{\end{eqnarray*}}

\newcommand{\X}{{\pmb {X}}}
\newcommand{\Y}{{\pmb Y}}

\newcommand{\tp}{\texttt{T}}
\newcommand{\mS}{\mathcal{S}}

\newcommand{\mA}{\pmb{A}}

\def\lv{\left\lvert}
\def\rv{\right\lvert}

\begin{document}
	\title{Robust Outlier Bound Condition to Phase Retrieval with Adversarial Sparse Outliers \footnote{This work was supported in part by the NSFC under grant numbers U21A20426, 12071426 and the National Key Research and Development Program of China under grant number 2021YFA1003500, and in part by the fundamental research funds for the central universities.}}

\author{Gao Huang\footnote{ School of Mathematical Science, Zhejiang University, Hangzhou 310027, P. R. China, E-mail address: hgmath@zju.edu.cn}}	
\author{Song Li\footnote{ School of Mathematical Science, Zhejiang University, Hangzhou 310027, P. R. China, E-mail address: songli@zju.edu.cn}}	
\author{Hang Xu\footnote{ Corresponding Author, School of Physics, Zhejiang University, Hangzhou 310027, P. R. China, E-mail address: hangxu@zju.edu.cn}}
\affil{}
\renewcommand*{\Affilfont}{\small\it}
	
	\date{}
	
	\maketitle
	\begin{abstract}
We consider the problem of recovering an unknown signal $\pmb{x}_0\in \mathbb{R}^{n}$ from phaseless measurements.  
In this paper, we study the convex phase retrieval problem via PhaseLift from linear Gaussian measurements perturbed by $\ell_{1}$-bounded noise and sparse outliers that can change an adversarially chosen $s$-fraction of the measurement vector. We show that the Robust-PhaseLift model can successfully reconstruct the ground-truth up to global phase for any $s\textless s^{*}\approx 0.1185$ with $\mathcal{O}(n)$ measurements, even in the case where the sparse outliers may depend on the measurement and the observation.
The recovery guarantees are based on the robust outlier bound condition and the analysis of the product of two Gaussian variables.
Moreover, we construct adaptive counterexamples to show that the Robust-PhaseLift model fails when $s\textgreater s^{*}$ with high probability.

			\end{abstract}
		

\section{Introduction}

Many problems in signal recovery can be treated as recovering an unknown discrete signal $\pmb{x}_0 \in \mathbb{R}^{n}$ or $\mathbb{C}^{n}$ from the following quadratic equations
\begin{equation}
		b_{i}=\vert\langle \pmb{a}_{i},\pmb{x}_0 \rangle\vert^{2},\quad i=1,\cdots,m. \label{pr_model}
	\end{equation} 
Here, $\pmb{a}_i \in \mathbb{R}^{n}$ or $\mathbb{C}^{n}$ and $b_i \in \mathbb{R}$ are given.
It stems from the practical problem that we can only obtain the squared modulus of the radiation or diffraction 
	pattern from the object but lose the phase from \eqref{pr_model} by using the current instruments. 
	If we cannot get the phase of the 
	light wave, a lot of information about the object or the light field will be lost.
Such problem is arisen from applications including X-ray crystallography \cite{millane1990phase,harrison1993phase}, quantum mechanics \cite{reichenbach1998philosophic} and quantum information \cite{heinosaari2013quantum}. Besides, several applications of phase retrieval also occur in many fields of image science, such as diffraction imaging \cite{bunk2007diffractive}, optics \cite{walther1963question}, astronomy \cite{fienup1987phase} and microscopy \cite{miao2008extending}. 
	
In the last decade, \cite{chai2010array,candes2015phase1,candes2013phaselift} introduced lifting-based algorithms which were named as PhaseLift for the phase retrieval problem \eqref{pr_model} and convexified the problem by lifting it into the space of matrices. The idea is to search for the lifted rank-$1$ matrix $\pmb{X}_0 = \pmb{x}_0\pmb{x}_0^{\tp}$, instead of the vector $\pmb{x}_0$. Thus, the phase retrieval problem \eqref{pr_model} can be cast as
	\begin{equation} \label{phaselift_normal}
		\begin{array}{ll}
		\text{minimize}   & \quad \operatorname{Tr(\pmb{X})}\\
		\text{subject to} & \quad  \pmb{X}\succeq 0\\
		                  & \quad \operatorname{Tr}(\pmb{a}_{i}\pmb{a}_{i}^{\tp}\pmb{X})=b_{i},\quad i=1,\cdots,m.\\
		\end{array}
		\end{equation}	
Let $\mathcal{A}: \mathcal{S}^{n }  \rightarrow  \mathbb{R}^{m}$ be the linear transformation $\mathcal{A}(\pmb{X})  = \sum_{i=1}^{m} \operatorname{Tr}(\pmb{a}_{i}\pmb{a}_{i}^{\tp}\pmb{X})\pmb{e}_{i}$,
	where $\pmb{e}_{1},\ldots,\pmb{e}_{m}$ denote the standard bases of $\mathbb{R}^{m}$ and $\mathcal{A}^{*}: \mathbb{R}^{m}  \rightarrow  {S}^{n }$ be the adjoint operator $\mathcal{A}^{*}(\pmb{y})   =  \sum_{i=1}^m y_{i}\pmb{a}_{i}\pmb{a}_{i}^{\tp}$.
	Formally, \eqref{phaselift_normal} can be expressed succinctly as the low-rank matrix recovery problem with rank-one measurements:
	\begin{equation}\label{m}
		\begin{array}{ll}
		\text{minimize}   & \quad \operatorname{Tr(\pmb{X})}\\
		\text{subject to} & \quad  \pmb{X}\succeq 0\\
		                  & \quad \mathcal{A}(\pmb{X})=\pmb{b}.\\
		\end{array}
		\end{equation} 
It is a semidefinite program (SDP). 
In \cite{candes2014solving}, E. Cand\`es et al. provided that the convex program recovers $\pmb{x}_0$ exactly up to global phase on the order of $n$ in both cases where $\pmb{x}_0$ is complex or real-valued when $\{\pmb{a}_i\}_i$ are i.i.d. drawn from uniform distribution or standard Gaussian distribution.
It can also be considered as the optimizationless recovery problem: one of finding a PSD 
matrix $\pmb{X}\succeq 0$ to encounter the measurements $\mathcal{A}(\pmb{X})=\pmb{b}$, which can be cast as a variant of original PhaseLift. Optimizationless recovery was firstly studied in \cite{demanet2014stable}.

In many applications, we do not get the noiseless observation but the data in the presence of noise:
\bea
b_{i}=\vert\langle \pmb{a}_{i},\pmb{x}_0 \rangle\vert^{2} + \omega_i,
\eea
where $\omega_i \in \mathbb{R}$ represent arbitrary measurement noise for all $i=1,\cdots,m$. Interested readers can refer to \cite{shechtman2015phase} for a comprehensive review on the scientific and
engineering background about this model. \cite{candes2013phaselift,candes2014solving, eldar2014phase} gave the stable analysis in the presence of  bounded noise $\pmb{\omega}$.
In \cite{candes2014solving}, the authors focused on the $\ell_1$-minimization:
\be\label{l1_min}
\begin{array}{ll}
		\text{minimize}   & \quad \Vert\mathcal{A}(\pmb{X}) - \pmb{b} \Vert_1\\
		\text{subject to} & \quad  \pmb{X}\succeq 0,
		\end{array}
\ee
which is the PhaseLift variant \cite{candes2014solving} and known as the Robust-PhaseLift model \cite{li2016low}. The authors in \cite{candes2014solving}  proved that the solution $\widetilde{\X}$ of the above program satisfies
\bea
\Vert \widetilde{\X} - \pmb{x}_0\pmb{x}_0^{\tp}\Vert_F \leq C\frac{\Vert \pmb{\omega}\Vert_1}{m}
\eea
with high probability. Here, $\pmb{\omega}$ is assumed to be bounded so that $\Vert \pmb{\omega}\Vert_1$ is finite. Besides, \cite{gao2016Stable,kueng2017low,kabanava2016stable,krahmer2022robustness} also refer to noise analysis related to phase retrieval.

In this paper, we focus on robust recovery of $n$-dimension signal from phaseless measurements \eqref{pr_model} possibly corrupted by bounded noise or adversarial sparse outliers:
\begin{equation}\label{formu}
		\pmb{b}=\mathcal{A}(\pmb{X}_0)+\pmb{\omega}+\pmb{z},\quad \pmb{X}_0=\pmb{x}_0\pmb{x}_0^{\tp},\quad \|\pmb{z}\|_{0}\leq sm.
	\end{equation}
Here, 
the sensing vectors $\{ \pmb{a}_i \}_i$ in $\mathcal{A}$ are all composed of i.i.d. standard Gaussian variables, 
$\pmb{\omega}$ denotes the $\ell_{1}$-bounded noise and $\pmb{z}$ represents the adversarial sparse outliers whose
nonzero entries can be arbitrary large. The fraction of nonzero entries is defined as $s:=\|\pmb{z}\|_{0}/m$.
 Adversarial sparse outliers are very common and inevitable in industrial applications, including sensor calibration \cite{li2016low}, face recognition \cite{de2003framework} and video surveillance \cite{li2004statistical}. 
We focus on the real-valued case and establish that the $\ell_1$-minimization can recover the signal $\pmb{x}_0$ (up to global phase) under $s$-fraction of adversarial sparse outliers where $s<s^{*}\approx0.1185$, as long as the number of measurements exceeds $\mathcal{O}\left(n\right)$, which is nearly optimal.
 
Some scholars
studied the robustness of the PhaseLift variant \eqref{l1_min} to gross, arbitrary outliers and raised that it can reconstruct the signal (up to global phase) with high probability provided only $\mathcal{O}(n)$ measurements \cite{hand2017phaselift, li2016low}. 
Different from the setting of sparse arbitrary outliers\footnote{Sparse arbitrary outliers are usually selected independently  which are assumed to be independent of $\{ \pmb{a}_i \}_i$.},
adversarial sparse outliers even highly depend on both the measurement $\mathcal{A}$ and the observation $\pmb{b}$.
Besides, exploring a tight threshold of adversarial sparse outliers is vitally critical in the field of signal recovery. A series of papers were posted to study this problem in decoding \cite{dwork2007price}, compressed sensing \cite{karmalkar2019compressed}
 and low-rank matrix recovery \cite{xu2022low}. However, up to now, the best fraction of adversarial sparse outliers is still not obtained in robust phase retrieval. P. Hand in \cite{hand2017phaselift} stated that the fraction cannot extend to $s \geq \frac{1}{2}$ as one could build a problem where half of the measurements are due to an $\pmb{x}_1$ and the other half to $\pmb{x}_2$.
However, he only proved that the PhaseLift variant \eqref{l1_min} can tolerate noise and a small, fixed fraction of gross outliers.
Y. Chi et al. in \cite{li2016low} restricted the signs of sparse noise to be generated from Rademacher distribution and extended the performance guarantee in  \cite{hand2017phaselift} to the general low-rank PSD matrix case. But they did not  give an explicit threshold of the fraction of sparse noise.
This raises a natural and important question in phase retrieval what fraction of adversarial sparse outliers can be admitted for recovery from the Robust-PhaseLift model \eqref{l1_min}. 
We show that the model can tolerate bounded noise and an $s$-fraction of adversarial sparse outliers with a threshold $s < s^{*}\approx 0.1185$, which partially responds to the question. To the best of our knowledge, this is the first theoretical performance guarantee to give an explicit threshold of sparse noise including adversarial sparse outliers in the robust phase retrieval.

The main breakthrough point in this paper is to consider the robust outlier bound condition (ROBC) and directly analyze the element $\pmb{x}\pmb{y}^{\tp}+\pmb{y}\pmb{x}^{\tp}$ in $T= \{\pmb{X} = \pmb{x} \pmb{y}^{\tp} + \pmb{y} \pmb{x}^{\tp}:\pmb{x}, \pmb{y} \in \mathbb{R}^n\}$.
In order to deal with the adversarial sparse outliers which may highly depend on both the measurement $\mathcal{A}$ and the observation $\pmb{b}$, we firstly study ROBC when $\{ \pmb{a}_i \}_i$  are all standard Gaussian vectors. It characterizes the upper and lower bound associated with partial measurements. 
It is worth noting that each entry of the measurement matrix $\pmb{A}_i$ in the sampling operator $\mathcal{A}$ is regarded as the product 
of two Gaussian variables (with the correlation coefficient $\rho=0$ or $1$) but no longer mutually independent after lifting\footnote{For example, we consider $\pmb{A}_i=\begin{bmatrix}
		a_1^2 &a_1a_2\\
		a_2a_1 &a_2^2
	\end{bmatrix} \in \mathbb{R}^{2\times 2}$
where $a_{1},a_{2} \overset{i.i.d.}{\sim}\mathcal{N}(0,1)$, and $a_1^2$ is not independent with $a_1a_2$.}.
To analyze the lower bound of $\| \mathcal{A}(\pmb{X})\|_{1}$ for $T$, E. Cand\`es et al. in \cite{candes2013phaselift} focused on all the symmetric rank-$2$ matrices which contains $T$, rather than considered directly the subset $T$, and
introduced random variable $|X_{1}^{2}-tX_{2}^{2}|$, where $X_{1}$ and $X_{2}$ are orthonormal normal random variables and $t\in[-1,1]$. 
In order to obtain the tighter lower bound in the analysis, we directly focus on the subset $T$ and analyze the product of two normal random variables\footnote{Here, we set $Z = X \cdot Y$ where $X {\sim} \mathcal{N}(0, 1)$, $X {\sim} \mathcal{N}(0, 1)$ and $Cov(X, Y) = \rho$.}.
The exact distribution of $Z = X \cdot Y$ has been studied since 1936 \cite{craig1936frequency} and deriving a closed-form expression for the exact Probability Density Function (PDF) of $Z$ for cases except  $\rho = 0, 1$ had been an open question for several decades. Recently, an approach based on characteristic functions has been used by S. Nadarajah and T. Pog\'any to obtain an explicit formula for the PDF of $Z$ \cite{nadarajah2016distribution}. It is more complicated to consider this variable $|Z|$ in the robust analysis rather than the absolute value of Gaussian variable in the former research about compressed sensing and low-rank matrix recovery \cite{karmalkar2019compressed, xu2022low}. 

In this paper, we firstly use the balance function to define the deviation between the $\ell_1$-norm of largest $s$-fraction and smallest $(1-s)$-fraction with respect to $|Z|$ and obtain the balance point $s^*\approx 0.1185$ when balance function takes the value $0$. We then show that a robust version of outlier bound condition holds for any $s$-fraction of the measurement $\mathcal{A}$ with high probability provided that $m \geq Cn$ when $\{ \pmb{a}_i \}_i$  are all standard Gaussian vectors, where $s < s^*$. Following literatures \cite{candes2013phaselift,candes2014solving,candes2015phase,gross2015partial,gross2017improved}, we search for an  inexact dual certificate which can fully approximate the exact one and has stronger restrictions on coefficients. We
focus on the model of \cite{li2016low} and obtain that the $\ell_1$-minimization is robust to any fraction of outliers $s < s^* $ by combining with the $s$-ROBC, even in the adversarial outlier case. 
We also construct adaptive counterexamples to show that the recovery fails when the fraction of outliers exceeds $s^*$.
And we corroborate our results through numerical experiments to increase the readability of this article. 

Throughout the paper, we use the following notations.	
For any matrix $\X$, we use $\sigma_{i}(\pmb{X})$ to denote its $i$-th largest singular value. We define $\|\pmb{X}\|_{p} = \Bigl[\sum_{i} \sigma_{i}^p(\pmb{X})\Bigr]^{1/p}$, hence 
$\| \pmb{X}\|_{1}$(or $\Vert \X \Vert_*$) is the nuclear norm of the matrix $\pmb{X}$, 
$\|\pmb{X}\|_2$(or $\| \pmb{X}\|_{F}$)  is the Frobenius norm
and $\|\pmb{X}\|_\infty$ (or $\| \pmb{X}\|$) is the operator norm. 
Besides, we define the inner product as 
$\langle\pmb{X}, \pmb{Y}\rangle=\operatorname{Tr}(\pmb{X} \pmb{Y}^{\tp})$ for any matrices $\X$ and $\Y$ with the same dimension.
Let $T_{\pmb{x}}$ denote the set of symmetric  tangent space of the form
$T_{\pmb{x}} = \{\pmb{X} = \pmb{x} \pmb{y}^{\tp} + \pmb{y} \pmb{x}^{\tp}:\pmb{y} \in \mathbb{R}^n\} $
and denote $T_{\pmb{x}}^\perp$ by its orthogonal complement, for convenience we will use $T$ and $T^\perp$ below. 
For two non-negative real sequences $\{a_t\}_t$ and $\{b_t\}_t$, we write $b_t = O(a_t)$ (or $b_t \lesssim a_t$) if $b_t \leq Ca_t$.
The sign function $\text{sgn}(\cdot)$ is defined as $\text{sgn}(x) = x/|x|$ if $x \neq 0$, and $\text{sgn}(0) = [-1, 1]$.

The rest of this paper is organized as follows. 
The main results are presented in Section \ref{mainresults}. In this section, we provide recovery guarantees based on the Robust-PhaseLift model. In Section \ref{robustobc}, we introduce the notion of robust outlier bound condition (ROBC) and show that ROBC holds in phase retrieval with high probability when $\{ \pmb{a}_i \}_i$  are all standard Gaussian vectors. The dual certificate is presented to play a crucial part in the proof in Section \ref{dualcertificate}. In Section \ref{prfs}, we provide the proofs for our main results. And a counterexample is constructed in Section \ref{prf_thm3} to prove Theorem \ref{theorem3} and it supports that the recovery fails when the fraction of outliers exceeds $s^*$. Proofs of auxiliary results are presented in \ref{appen_prfs}. We also corroborate our theory through numerical experiments in \ref{experiments}.

\section{Main Results}\label{mainresults}

In this section, we present our main theoretical results on Robust-PhaseLift. Recalling that the noisy measurements are
\begin{equation*}
		\pmb{b}=\mathcal{A}(\pmb{X}_0)+\pmb{\omega}+\pmb{z},\quad \pmb{X}_0=\pmb{x}_0\pmb{x}_0^{\tp},\quad \|\pmb{z}\|_{0}\leq sm,
	\end{equation*}
we focus on the the Robust-PhaseLift model \eqref{l1_min} and set the support of adversarial sparse outlier $\pmb{z}$ to be drawn from the Rademacher distribution as $\mathbb{P}\{\text{sgn}(z_i) = -1\} = \mathbb{P} \{\text{sgn}(z_i) = 1\} = {1\over 2}$ for each $i \in \text{supp}(\pmb{z})$.
This setting is illuminated by \cite{li2016low} to make dual certificates sufficiently accurate. We prove that the model can tolerate bounded noise and an $s$-fraction of adversarial sparse outliers with a threshold $s < s^{*}\approx 0.1185$. To the best of our knowledge, this is the first theoretical performance guarantee to give an explicit threshold of adversarial sparse outliers in the robust phase retrieval.
	\begin{theorem} \label{theorem}
	Suppose that $s\textless s^{*}-2\gamma$ and $s=\Vert\pmb{z}\Vert_{0}/m$. Here $s^{*}\approx 0.1185$ and the number of measurements satisfies $m\geq C[\gamma^{-2}\ln\gamma^{-1}]n$, where $C$ is a sufficiently large constant. Assume the support of $\pmb{z}$ is fixed with the signs of its nonzero entries generated from the Rademacher distribution.
If the $\ell_1$-bounded noise $\pmb{\omega}=\pmb{0}$, then for all $\pmb{x}_0\in \mathbb{R}^{n}$ the solution $\widehat{\pmb{X}}$ of (\ref{l1_min}) satisfies
\begin{eqnarray*}
		\widehat{\pmb{X}}=\pmb{x}_{0}\pmb{x}^{\tp}_{0}
	\end{eqnarray*}
with probability at least $1-\mathcal{O}(e^{-cm\gamma^{2}})$, where $c$ is a positive numerical constant. 
\end{theorem}
	
		Theorem \ref{theorem} gives the theoretical guarantee of phase retrieval only with adversarial sparse outliers. Then we post Theorem \ref{theorem2} for the analysis with the mixed noises.
		\begin{theorem}\label{theorem2}
		Suppose that $s\textless s^{*}-2\gamma$ and $s=\Vert\pmb{z}\Vert_{0}/m$, where $s^{*}\approx 0.1185$. Assume the support of $\pmb{z}$ is fixed with
		the signs of its nonzero entries generated from the Rademacher
		distribution. 
		When the number of measurements satisfies $m\geq C[\gamma^{-2}\ln\gamma^{-1}]n$, where $C$ is a sufficiently large
		constant, then there exist some positive numerical constants $C_{1}$, $C_{2}$ and $c$ such that the solution $\widehat{\pmb{X}}$ of (\ref{l1_min}) satisfies
		\begin{eqnarray*}
			\left\|\widehat{\pmb{X}}-\pmb{x}_{0}\pmb{x}^{\tp}_{0}\right\|_{F}\leq \left(\frac{C_{1}}{\gamma}+C_{2}\right)\frac{\|\pmb{\omega }\|_{1}}{m}
		\end{eqnarray*}
		with probability at least  $1-\mathcal{O}(e^{-cm\gamma^{2}})$ for all $\pmb{x}_{0}\in \mathbb{R}^{n}$.
Furthermore, by finding the largest eigenvector with largest eigenvalue of $\widehat{\pmb{X}}$, we can construct an estimate obeying: 
\begin{eqnarray*}
	\min_{k\in\{0,1\}} \left\|\hat{\pmb{x}}-(-1)^{k}\pmb{x}_{0} \right\|_{2} \leq \left(\frac{C_{1}}{\gamma}+C_{2} \right) \min \left( \|\pmb{x}_{0}\|_{2}, \frac{ \|\pmb{\omega }\|_{1} }{ m\|\pmb{x}_{0}\|_{2} }\right).
\end{eqnarray*}
		\end{theorem}
	\begin{remark}
	In \cite{hand2017phaselift}, the author also focused on the phase retrieval problem with mixed noises, and the authors in \cite{li2016low} extended the performance guarantees in \cite{hand2017phaselift} to the general low-rank PSD matrix case. They considered the existence and feasibility of sparse outlier $\pmb{z}$ and pointed out that the program (\ref{l1_min}) is additionally robust against a constant fraction of arbitrary errors. However, they did not give an explicit threshold of the fraction of $\pmb{z}$. 
	
	\end{remark}
We are inspired by the ideas in \cite{li2016low} and focus on the model \eqref{formu} with the signs of sparse outlier generated from the Rademacher distribution. In practice, some noises are connected with this distribution, such as salt-and-pepper noise \cite{tu2023new} and zero-mean noise. These occur in a number of applications such as image denoising \cite{ulyanov2018deep}, low-light image enhancement \cite{li2015low}. Here, we introduce the Rademacher distribution only for the purpose of constructing the dual certificates. 
Without this assumption, we need to restrict $|\mS|=\mathcal{O}(n)$
in the analysis of dual certificates \cite{hand2017phaselift}, which means that the support of sparse outliers cannot be sufficiently small. 
Besides, 
our threshold $s^{*}$ is only based on the robust outlier bound condition (ROBC) and has no connection with the Rademacher distribution. Specifically, it comes from the analysis of the 
product of two normal random variables. 
Thus, we conjecture that  $s<s^*\approx 0.1185$ is a tight threshold for any adversarial sparse outliers even when the Rademacher distributional assumption is taken off.

In the following theorem, we will illustrate that $\ell_1$-minimization  almost fails to exactly recover the matrix $\pmb{X}_0=\pmb{x}_0\pmb{x}_0^{\tp}$ when the fraction $s>s^*\approx 0.1185$.
\begin{theorem} \label{theorem3}
	Set $s\textgreater s^{*}+2\gamma$ where $s^{*}\approx 0.1185$. In the case $\pmb{\omega}=\pmb{0}$, for any $\pmb{x}_0\in \mathbb{R}^{n}$, there exist  an adversarial
	sparse outlier $\pmb{z}$ with $\Vert \pmb{z}\Vert_{0}=sm$, such that the solution of (\ref{l1_min}) is not exactly the ground-truth $\pmb{X}_0=\pmb{x}_0\pmb{x}_0^{\tp}$ with probablity at least $1-\mathcal{O}(e^{-cm\gamma^{2}})$.	
\end{theorem}

And in Section \ref{prf_thm3}, we construct a counterexample that the sparse noise vector $\pmb{z}$ depends on $\mathcal{A}$ and $\pmb{b}$. We will also construct numerical experiments in  Section \ref{sec_6_2} to show the recovery miscarriages when the fraction of corruptions exceeds $s^{*}$. 

\begin{remark}
We point out that experiments by P. Hand \cite{hand2017phaselift} showed the reconstructions succeed in the presence of sparse uniform noise ($z_i \sim \text{Uniform}([0, 10^4])$ if $i\in \mS$) with noise rate $0.05$, while
experiments by Y. Chi et al. \cite{li2016low} showed the reconstructions still succeed when the measurements are selected uniformly at random and corrupted by standard Gaussian variables with noise rate exceeding $0.2$.
However, in experiments the noises were not chosen adaptively, that is why the reconstructions still succeed when $s = 0.2$.
\end{remark}

\section{Robust Outlier Bound Condition}\label{robustobc}
The mixed-norm $\ell_1/\ell_2$-restricted isometry property ($\ell_1/\ell_2$-RIP) has been usually used in the context of structured signal recovery in the previous literature \cite{donoho2006most,chen2015exact,cai2015rop,allen2016restricted}.

\begin{definition}[$\ell_1/\ell_2$-RIP]
An operator $\mathcal{A}: \mathbb{R}^{n_1\times n_2} \to \mathbb{R}^{m}$ is said to satisfy the mixed-norm $\ell_1/\ell_2$-restricted isometry property, if there exists a constant $\delta_{2r}$ such that
\bea
(1-\delta_{2r}) \left\Vert \pmb{X} \right\Vert_F \leq \frac{1}{m}\left\Vert \mathcal{A}(\pmb{X}) \right\Vert_1 \leq (1+\delta_{2r}) \left\Vert \pmb{X} \right\Vert_F
\eea
holds for all rank-$2r$ matrices $\pmb{X} \in \mathbb{R}^{n_1\times n_2}$. The RIP constant $\delta_{2r}$ is the smallest constant satisfying the above inequality.
\end{definition}
	
	Let $\left[ m \right]$ denote the set $\{1,2,\cdots ,m\}$, $\mS$ be a subset of $\left[ m \right]$ and $\mS^{c}$ be the complement set of $\mS$ with respect to $\left[ m \right]$. Let $\mathcal{A}_{\mS}$ be the mapping operator $\mathcal{A}$ constrained on $\mS$, which is defined as
	$[\mathcal{A}_{\mS}(\pmb{X})]_{i} = \pmb{a}_{i}^{\tp}\pmb{X}\pmb{a}_{i}$ if $i \in \mS$ and $[\mathcal{A}_{S}(\pmb{X})]_{i}= 0$ if $i \in \mS^c$. If the support of the sparse outlier $\mS$ is known as prior knowledge, the mixed-norm $\ell_1/\ell_2$-RIP can lead to the restricted outlier bound condition \cite{duchi2019solving,charisopoulos2021low, xu2022CONVERGENCE}:
	\bea
	\frac{1}{m}\left\Vert \mathcal{A}_{\mS^c}(\pmb{X}) \right\Vert_1 - \frac{1}{m}\left\Vert \mA_{\mS}(\pmb{X})\right\Vert_1 \geq C(\delta) \left\Vert \pmb{X}\right\Vert_F.
	\eea
	Here, we need the partial measurements $\mA_{\mS^c}$, $\mA_{\mS}$ to satisfy the $\ell_1/\ell_2$-RIP when the support $\mS$ of sparse outlier is fixed \cite{hand2017phaselift, li2016low}.
	
	In order to deal with the adversarial sparse outliers whose support may probably depend on the observation $\pmb{b}$ and the measurement operator $\mathcal{A}$, we use the following $s$-robust outlier bound condition to characterize the upper and lower bound associated with partial measurements.

	\begin{definition}[$s$-robust outlier bound condition]\label{l_3}
		An operator $\mathcal{A}: \mathbb{R}^{n_1\times n_2} \to \mathbb{R}^m$ is said to satisfy the $s$-robust outlier bound condition, if there exists a constant $C(s) > 0$ depending on $s\in [0,1]$ such that the following holds for all matrix $\pmb{X}\in 	T$,
		\begin{eqnarray}\label{rip1}
			\frac{1}{m}		\min_{\mS \subset[m],  |\mS|\leq s m} \Bigl[\Vert \mathcal{A}_{\mS^{c}}(\pmb{X})\Vert_1 -	 \Vert \mathcal{A}_\mS(\pmb{X})\Vert_1\Bigr]&\geq&  C(s)\cdot\Vert \pmb{X}\Vert_F.
		\end{eqnarray}
		Here, $T$ is a subset of $\mathbb{R}^{n_1\times n_2}$.
	\end{definition}
\subsection{Product of Normal Random Variables}
We denote $T= \{\pmb{X} = \pmb{x} \pmb{y}^{\tp} + \pmb{y} \pmb{x}^{\tp}:\pmb{x}, \pmb{y} \in \mathbb{R}^n\} $ in this section.
Under the setting of \eqref{m}, to analyze the lower bound of $\| \mathcal{A}(\pmb{X})\|_{1}$, E. Cand\`es et al. in \cite{candes2013phaselift} focused on all the symmetric rank-$2$ matrices which contains $T$, rather than considered directly the subset $T$.
Since any rank-$2$ matrix with spectral norm-$1$ has eigenvalue decomposition $\pmb{X}=\pmb{u}_{1}\pmb{u}^{\tp}_{1}-t\pmb{u}_{2}\pmb{u}^{\tp}_{2}$ with $t\in[-1,1]$ and $\langle\pmb{u}_{1},\pmb{u}_{2}\rangle =0$, then 
\begin{eqnarray*}
	\| \mathcal{A}(\pmb{X})\|_{1}= 
	 \sum_{i=1}^{m}\left\lvert| \langle\pmb{u}_{1}, \pmb{a}_{i}\rangle|^{2} -t| \langle\pmb{u}_{2}, \pmb{a}_{i}\rangle|^{2}  \right\rvert 
	=  \sum_{i=1}^{m}|X_{1}^{2}-tX_{2}^{2}|,
\end{eqnarray*}
where $X_{1}$ and $X_{2}$ are orthonormal normal random variables.

In order to obtain the tighter lower bound in the analysis, we focus on the subset $T$ and introduce another random variable --- product of two normal random variables.
For any $\pmb{X} = \pmb{x} \pmb{y}^{\tp} + \pmb{y} \pmb{x}^{\tp} \in T$, we have
\begin{eqnarray}
		\| \mathcal{A}(\pmb{X})\|_{1}&= &\| \mathcal{A}(\pmb{x}\pmb{y}^{\tp}+\pmb{y} \pmb{x}^{\tp})\|_{1}
		= \sum_{i=1}^{m}\lvert \pmb{a}_{i}^{\tp}\pmb{x} \pmb{y}^{\tp}\pmb{a}_{i} + \pmb{a}_{i}^{\tp}\pmb{y} \pmb{x}^{\tp}\pmb{a}_{i}  \rvert\nonumber\\
		&= & \sum_{i=1}^{m}\left\lvert \langle\pmb{x}, \pmb{a}_{i}\rangle  \langle\pmb{a}_{i}, \pmb{y}\rangle+   \langle\pmb{y}, \pmb{a}_{i}\rangle  \langle\pmb{a}_{i}, \pmb{x}\rangle \right\rvert
		= 2\sum_{i=1}^{m}\lvert \langle\pmb{x}, \pmb{a}_{i}\rangle  \langle\pmb{y}, \pmb{a}_{i}\rangle \rvert. \label{XtimesY}
	\end{eqnarray}
	
From the above equation, we can observe that when $\pmb{x}$ and $\pmb{y}$ are fixed with $\Vert \pmb{x}\Vert_2 = \Vert \pmb{y}\Vert_2 =1$, $\pmb{a}_{i}\sim \mathcal{N}(0,\pmb{I}_{n}) $ yields $ X=\langle\pmb{x}, \pmb{a}_{i}\rangle{\sim} \mathcal{N}(0, 1)$ and $Y=\langle\pmb{y}, \pmb{a}_{i}\rangle{\sim} \mathcal{N}(0, 1)$.
Besides,
	the correlation coefficient between random variables $X$ and $Y$ is	
	\begin{eqnarray*}
		\rho_{XY}&=&\frac{Cov(X,Y)}{\sqrt{Var(X)}\sqrt{Var(X)}}
		=\mathbb{E}\left[X \cdot Y\right]
	=	\mathbb{E}\left[\langle\pmb{x}, \pmb{a}_{i}\rangle\cdot\langle\pmb{y}, \pmb{a}_{i}\rangle\right]
		=\langle\pmb{x}, \pmb{y}\rangle=\rho.
	\end{eqnarray*}
	Thus, the random variable $\langle\pmb{x}, \pmb{a}_{i}\rangle  \langle\pmb{y}, \pmb{a}_{i}\rangle$ can be regarded as the product of two standard
	normal random variables $X,Y$ with the correlation coefficient $\rho$.	
	
	Let $Z_{\rho}=X\cdot Y$ be the random variable described above.
	The problem of looking for the exact distribution of	 $Z_{\rho}$ has a long history and can date back to 1936 \cite{craig1936frequency},
	which was not completely solved until 2016 \cite{nadarajah2016distribution,cui2016exact,gaunt2019note}.
	We directly give the PDF of $Z_{\rho}$. 
	\begin{lemma}[\cite{nadarajah2016distribution}]\label{pdf_xy}
Let $(X, Y)$ denote a bivariate normal random vector with zero means, unit variances and correlation coefficient $\rho$. Then, the PDF of $Z_{\rho} = X\cdot Y$ is 
\be \label{S}
f(z) = 
\begin{cases}
\frac{1}{\sqrt{-2\pi z}}e^{z/2} &\rho=-1 \\
 \frac{1}{\pi\sqrt{1-\rho^2}}e^{ \frac{\rho z}{1-\rho^2}}K_0\left( \frac{|z|}{1-\rho^2} \right) &-1<\rho<1\\
 \frac{1}{\sqrt{2\pi z}}e^{-z/2} &\rho=1
\end{cases} 
\ee
for all $-\infty < z < +\infty$, where $K_0(\cdot)$ denotes the modified Bessel function of the second kind of order zero. When $\rho\to\pm 1$, the PDF of $Z_{\rho}$ is asymptotic to the
PDF of $Z_{\pm1}$.
\end{lemma}
Bessel function of the second kind of order $\nu$ \cite{gaunt2019note} is 
\bea
K_{\nu}(x) = \int_{0}^{\infty} e^{-x \cosh(t)} \cosh(\nu t) dt,
\eea 
where $\cosh(t) = \frac{e^t+e^{-t}}{2}$, thus $K_0(x) = \int_{0}^{\infty} e^{-x \cosh(t)} dt$. From the lemma, we can know that when $\rho = 0$, $f(z) = \frac{1}{\pi} K_0\left( |z|\right)$.

\begin{figure}[H]
\centering 

\begin{minipage}[b]{0.45\textwidth} 
\centering 
\includegraphics[width=1\textwidth]{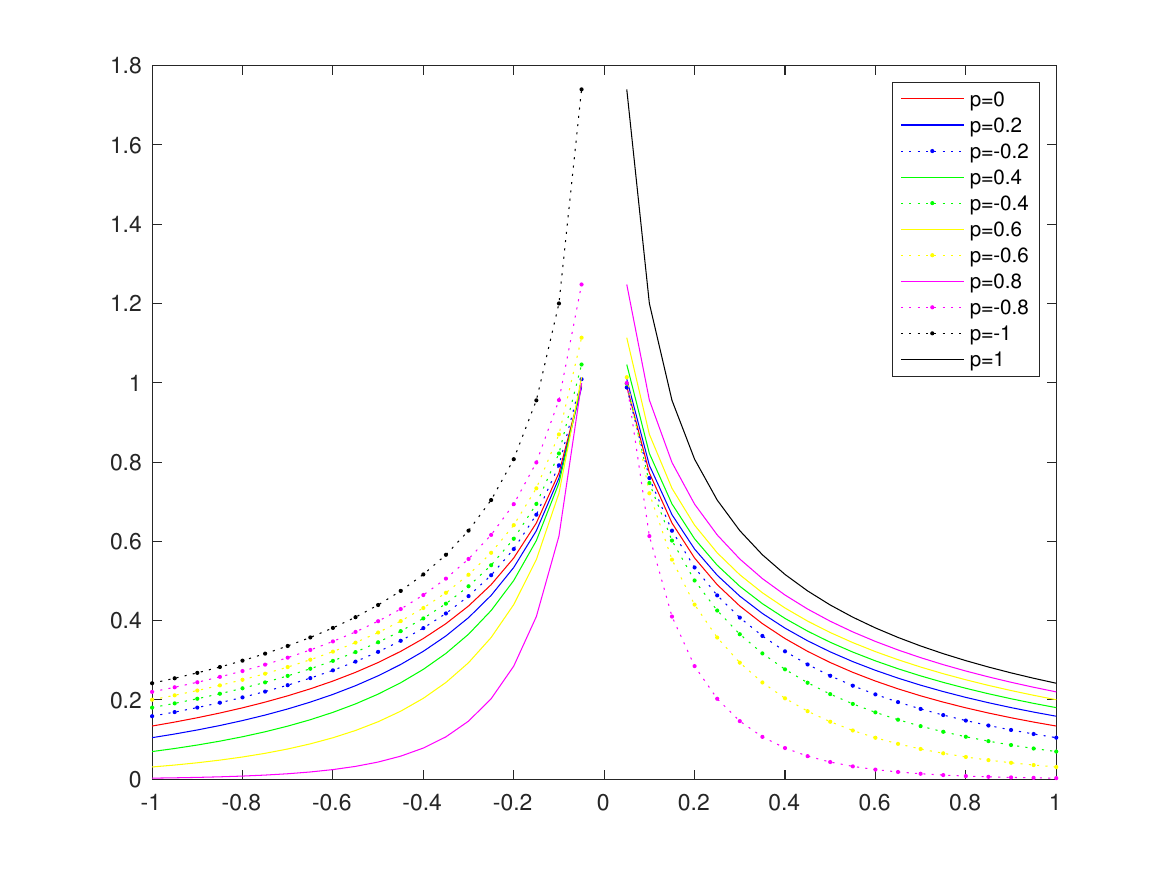}
\caption{PDF of $Z_{\rho}$ in $[-1, 1]$.}
\label{Figpr1}
\end{minipage}
\begin{minipage}[b]{0.45\textwidth} 
\centering 
\includegraphics[width=1\textwidth]{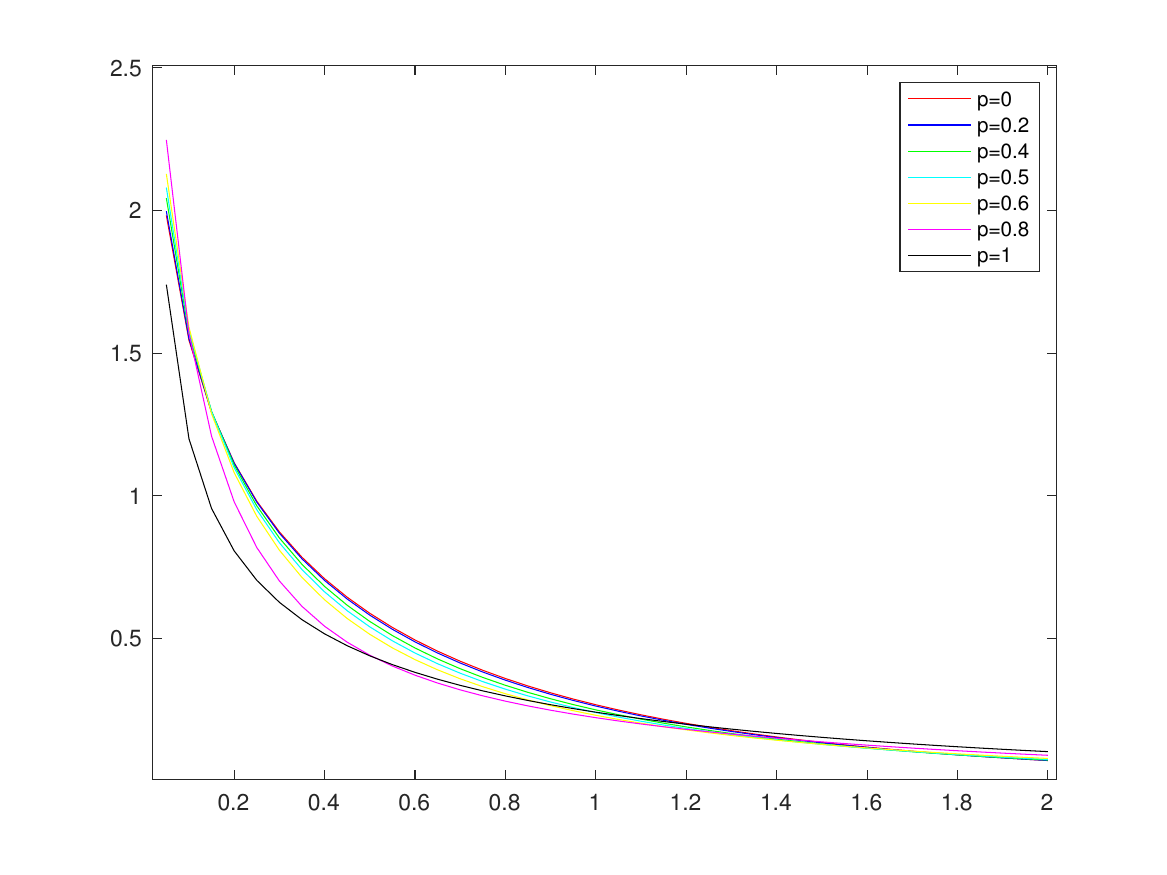}
  \caption{PDF of $|Z_{\rho}|$ in $[0, 2]$.}
\label{Figpr2}
\end{minipage}

\end{figure}

Since we need to analyze the $\ell_1$-norm of largest $s$-fraction and smallest $(1-s)$-fraction with respect to this distribution, the PDF of $|Z_{\rho}|$ is to be considered, which can be trivially obtained by Lemma \ref{pdf_xy}.

\begin{lemma}\label{pdf_xy_abs}
Let $(X, Y)$ denote a bivariate normal random vector with zero means, unit variances and correlation coefficient $\rho$. Then, the PDF of $|Z_{\rho}|$ is 
\be \label{S_abs}
f_{\rho}(z) = \begin{cases} 
\frac{1}{\sqrt{2\pi z}}e^{-z/2} &\rho=\pm 1 \\ 
 \frac{2}{\pi\sqrt{1-\rho^2}} \cosh \left(\frac{\rho z}{1-\rho^2}\right) K_0\left( \frac{z}{1-\rho^2} \right) &-1<\rho<1.
\end{cases} 
\ee
\end{lemma}
We set the Cumulative Distribution Function (CDF) of $|Z_{\rho}|$
\bea
F_{\rho}(t) = \int_{0}^{t} f_{\rho}(z) dz,
\eea
and define
\bea
H(\rho, s) = H_{1}(\rho, s)-H_{2}(\rho, s):=\left[ \int_0^{F_{\rho}^{-1}(1-s)} - \int_{F_{\rho}^{-1}(1-s)}^{+\infty} \right] zf_{\rho}(z)dz
\eea
as the balance function which is the deviation between the $\ell_1$-norm of largest $s$-fraction and smallest $(1-s)$-fraction the variable $|Z_{\rho}|$. With the above concepts, we can define the core concept of this paper for dissecting the properties of ROBC.
	\begin{definition}
	We call 	
	\begin{eqnarray}
		\mathcal{H}^{*}(s)=\min_{\rho \in[0,1]}\sqrt{\frac{2}{1+\rho^{2}}}\left[\int_{0 }^{F_{\rho}^{-1}(1-s)}- \int_{F_{\rho}^{-1}(1-s) }^{+\infty} \right]zf_{\rho }(z)dz
	\end{eqnarray}
	the minimum balance function.
	\end{definition}	
Since we consider the absolute value of $Z_{\rho}$ here, we restrict $\rho \in [0,1]$ for convenience by the property of symmetry.
Below we will reveal the properties of $\mathcal{H}^{*}(s)$.
\begin{proposition}\label{H}
The minimum balance function $\mathcal{H}^{*}(s)$ satisfies:
\begin{itemize}
	\item[$\mathrm{(a)}$]\label{aaaaa}
	$\mathcal{H}^{*}(0)=\min\limits_{\rho \in[0,1]}\sqrt{\frac{2}{1+\rho^{2}}}\mathbb{E}[|Z_{\rho}|]=\frac{2\sqrt{2}}{\pi}\,\text{ and }\,\mathcal{H}^{*}(1)=-\max\limits_{\rho \in[0,1]}\sqrt{\frac{2}{1+\rho^{2}}}\mathbb{E}[|Z_{\rho}|]=-1;$
	\item[$\mathrm{(b)}$]\label{H_c} $\mathcal{H}^{*}(s)$ is a strictly monotonically decreasing continuous function on interval $[0,1]$;
	\item[$\mathrm{(c)}$] There exists constant $L_{0}>0$ such that for any $0\le s_{1}\le s_{2}\le1$,
	\begin{eqnarray*}
	\mathcal{H}^{*}(s_{1})-\mathcal{H}^{*}(s_{2})\le L_{0}(s_{2}-s_{1});
	\end{eqnarray*}
	\item[$\mathrm{(d)}$] $\mathcal{H}^{*}(s)$ has a unique zero point where $s^*\approx 0.1185$ makes $\mathcal{H}^*(s^*) = 0$. And $\rho^* = 0.795$ makes $H(\rho^*,s^*)=0$. Besides, there exists a constant $l_0>0$ such that
	\bea
	\mathcal{H}^{*}(s^*-\gamma)-\mathcal{H}^{*}(s^*) \geq l_0 \gamma
	\eea
	holds for $\gamma \in (0, s^*)$.
\end{itemize}
\end{proposition}

\subsection{s-ROBC Constant $C(s)$}
	In this subsection, we will use $\mathcal{H}^{*}(s)$ and the product of two normal random variables to drive that the linear transformation $\mathcal{A}$ in \eqref{m} fulfills the ROBC with high probability and obtain the s-ROBC constant $C(s)$ correspondingly.
	\begin{theorem}\label{l_17}
		Assume $0\textless\gamma\textless s^{*}$ and $c_0$, $C_0$ are positive numerical constants. If $s<s^*$ and $m \geq C_{0}[\gamma^{-2}\ln\gamma^{-1}]n$, then then with probability at least $1- \mathcal{O}(e^{-c_{0}m\gamma^2})$, $\mathcal{A}$ satisfies the $s$-robust outlier bound condition, that is for all $\pmb{X}\in T$:
		\begin{eqnarray}
			\frac{1}{m}		\min_{\mS \subset[m],  |\mS|\leq s m} \Bigl[\Vert \mathcal{A}_{\mS^{c}}(\pmb{X})\Vert_1 -	 \Vert \mathcal{A}_{\mS}(\pmb{X})\Vert_1\Bigr]&\geq& \left[\mathcal{H}^{*}(s+\gamma)-\frac{l_{0}\gamma}{2}\right]\cdot\Vert \pmb{X}\Vert_F.
		\end{eqnarray}
	\end{theorem}	
\begin{remark}
The establishment of Theorem \ref{l_17} does not depend on the fixation of the support set $\mS$ and the assumption that the signs of sparse outlier's nonzero entries generated from the Rademacher
distribution. Besides, the condition $s\textless s^*$ is to ensure $\mathcal{H}^{*}(s+\gamma)-l_{0}\gamma/2 \geq 0$, otherwise the lower bound is meaningless.
\end{remark}
When we set $\mS=\emptyset$, a straightforward corollary of Theorem \ref{l_17} is Lemma 3.2 in \cite{candes2013phaselift} as following, which is an enhanced version in a sense. 
\begin{corollary}\label{co1}
		Assume $0\textless\gamma\textless s^{*}$ and $c_0$, $C_0$ are positive numerical constants. If $m \geq C_{0}[\gamma^{-2}\ln\gamma^{-1}]n$, $\mathcal{A}$ obeys the the following property with probability at least $1- \mathcal{O}(e^{-c_{0}m\gamma^2})$ :
		\begin{eqnarray}
			\frac{1}{m}	\Vert \mathcal{A}(\pmb{X})\Vert_1 &\geq& \frac{2\sqrt{2}}{\pi}(1-\gamma)\Vert \pmb{X}\Vert_F, \text{ for all } \pmb{X}\in T.
		\end{eqnarray}
		\end{corollary}
\begin{remark}
	There are two  foremost differences  between Corollary \ref{co1} and Lemma 3.2 in \cite{candes2013phaselift}. One is that we	restrict the lower bound to $\Vert\pmb{X}\Vert_{F}$ rather than $\Vert\pmb{X}\Vert$. Another is we premeditate all $\pmb{X}\in T$ instead of all symmetric rank-2 matrix $\pmb{X}$. Considering the $\Vert \cdot\Vert_F$ in Lemma 3.2 of \cite{candes2013phaselift}, we can get $\frac{1}{m}	\Vert \mathcal{A}(\pmb{X})\Vert_1 \ge \frac{0.94}{\sqrt{2}}(1-\gamma)\Vert \pmb{X}\Vert_F$ because $\Vert \pmb{X}\Vert_F\le\sqrt{2}\Vert \pmb{X}\Vert$ for any $\pmb{X}\in T$.  
	Since $\frac{2\sqrt{2}}{\pi}\textgreater\frac{0.94}{\sqrt{2}}$, the Corollary \ref{co1} is a tighter version.
\end{remark}
		The proof of Theorem \ref{l_17} is  rather lengthy. Our strategy is first to prove that $\mathcal{A}$ obeys ROBC for a fixed matrix $\pmb{X}$ and then use the standard covering-number based argument to extend to the result for all $\pmb{X} \in T$.
	
	\begin{lemma}\label{X}
	Fix $\pmb{X} = \pmb{x} \pmb{y}^{\tp} + \pmb{y} \pmb{x}^{\tp}$ and $\rho=\langle\pmb{x}/\|\pmb{x}\|_{2}, \pmb{y}/\|\pmb{y}\|_{2}\rangle$. For $0 \textless\gamma \textless s^{*}$, there exists positive numerical constant $c$ such that with probability at least $1-\mathcal{O}(e^{-cm\gamma^{2}})$,	
		\begin{eqnarray}
	\frac{1}{m}\min_{\mS \subset[m],  |\mS|\leq s m}\frac{\Vert \mathcal{A}_{\mS^{c}}(\pmb{X})\Vert_1 }{\Vert \pmb{X}\Vert_F} &\geq& \sqrt{\frac{2}{1+\rho^2}}H_{1}(\rho, s+\gamma)-\frac{l_{0}\gamma}{8}\label{l_X1},\\	
		\frac{1}{m}	\max_{\mS \subset[m],  |\mS|\leq s m}\frac{\Vert \mathcal{A}_{\mS}(\pmb{X})\Vert_1}{\Vert \pmb{X}\Vert_F} &\leq& \sqrt{\frac{2}{1+\rho^2}}H_{2}(\rho, s+\gamma)+\frac{l_{0}\gamma}{8}\label{l_X2}.
		\end{eqnarray}
	\end{lemma}
	\begin{proof}
		We divide the proof into three steps. The first step is to construct  the link between the $\ell_{1}$-norm of $\mathcal{A}(\pmb{X})$ and the random variable $\lvert  Z_{\rho} \rvert$. It is then stated that the indicator function $\lvert  Z_{\rho} \rvert\cdot \mathbbm{1}_{[0,t]}$ is sub-exponential and has a consistent sub-exponential norm \footnote{
For $X$, Orlicz norm $\|\cdot\|_{\psi_{1}}$ is defined as $\|X\|_{\psi_{1}}=\inf\{t>0: \mathbb{E} e^{|X|/t} \leq 2\}$, while $\|\cdot\|_{\psi_{2}}$ is defined as	$\|X\|_{\psi_{2}}=\inf\{t>0: \mathbb{E} e^{X^{2}/t^{2}} \leq 2\}$.}. 
The concentration step will be based on the Dvoretzky-Kiefer-Wolfowitz Inequality and Bernstein Inequality.
		
		{\bf Step 1: Convert to $\lvert  Z_{\rho} \rvert$.}
	We only need to consider all $\pmb{X}\in \widetilde{T} =\{\pmb{X} = \pmb{x} \pmb{y}^{\tp} + \pmb{y} \pmb{x}^{\tp}:\pmb{x}, \pmb{y} \in \mathcal{S}^{n-1}\} $ , as $\| \pmb{X}\|_{F}=\|(a\pmb{x})( b\pmb{y}^{\tp}) + (b\pmb{y} )(a\pmb{x}^{\tp})\|_{F} =\lvert a \rvert \lvert b \rvert  \|\pmb{x} \pmb{y}^{\tp} + \pmb{y} \pmb{x}^{\tp}\|_{F}$ and $	\| \mathcal{A}((a\pmb{x})( b\pmb{y}^{\tp}) + (b\pmb{y} )(a\pmb{x}^{\tp})) \|_{1}=\lvert a \rvert \lvert b \rvert 	\| \mathcal{A}( \pmb{x} \pmb{y}^{\tp} + \pmb{y} \pmb{x}^{\tp})\|_{1}$. 
		
		Now, let $\pmb{x}, \pmb{y} \in \mathcal{S}^{n-1}$, then $\rho=\langle\pmb{x}, \pmb{y}\rangle$. 
		Let $\pmb{x}=(x_{1},x_{2},\dots ,x_{n})^{\tp}$ , 	$\pmb{y}=(y_{1},y_{2},\dots ,y_{n})^{\tp}$.
		A direct calculation drives	
		\begin{eqnarray*}
			\|\pmb{X}\|_{F}^{2}&= &\sum_{1\leq i,j\leq n}(x_{i}y_{j}+x_{j}y_{i})^{2}
			=\sum_{1\leq i,j\leq n}(x_{i}^{2}y_{j}^{2}+x_{j}^{2}y_{i}^{2}+2x_{i}x_{j}y_{i}y_{j})\\
			&=&2\left(\sum_{i=1}^{n}x_{i}^{2}\right)\left(\sum_{i=1}^{n}y_{i}^{2}\right)+2\left(\sum_{1\leq i\leq n} x_{i}y_{i}\right)^{2}
			=2+2\rho^{2}.
		\end{eqnarray*}
		Further, from the previous description \eqref{XtimesY}, we have
		\begin{eqnarray}\label{l_T1}
		\frac{\| \mathcal{A}(\pmb{X})\|_{1}}{\|\pmb{X}\|_{F}}= \sqrt{\frac{2}{1+\rho^{2}}} \sum_{i=1}^{m}\lvert \langle\pmb{x}, \pmb{a}_{i}\rangle  \langle\pmb{y}, \pmb{a}_{i}\rangle \rvert= \sqrt{\frac{2}{1+\rho^{2}}} \sum_{i=1}^{m}\lvert Z_{\rho}\rvert:=\sum_{i=1}^{m}\xi_{i}.
	\end{eqnarray}
Here, since $\pmb{x}$, $\pmb{y}$ are fixed with $\Vert \pmb{x} \Vert_2 = \Vert \pmb{y}\Vert_2 =1$, the variable $\langle\pmb{x}, \pmb{a}_{i}\rangle  \langle\pmb{y}, \pmb{a}_{i}\rangle$ could be regarded as the product of two standard Gaussian variables $Z_{\rho} = X\cdot Y$ with the correlation coefficient $\rho$ varying from $[-1, 1]$.

		{\bf Step 2: Truncation.}
		Denote $g(x) = x \cdot \mathbbm{1}_{[0,t]}(x)$, and $\Xi = g(\xi)$ where $\xi$ is defined as above (\ref{l_T1}). 
		We will derive that for any $t\in \mathbb{R}^{+}\cup \{+\infty\}$, 
		 there 
exists $ K_{0}>0$ such that $	\| \Xi \|_{\psi_{1}} \leq K_{0}$.

	It is clear that for fixed $t$, $ \| \Xi \|_{\psi_{1}}< \infty $, which is due to $\Xi$ is a bounded random variable.
	When $t=+\infty $, $\Xi=\xi$ is a sub-exponential random variable by 
	\begin{eqnarray}\label{Z}
	\|\lvert Z_{\rho}\rvert\|_{\psi_{1}}\leq 	\|X\|_{\psi_{2}}^{2},
\end{eqnarray}
where $X\sim\mathcal{N}(0,1)$\footnote{Here, we give the proof of \eqref{Z} in \ref{ZP}. }
.
	We first consider  the mapping
	\begin{eqnarray*}
		\mathcal{F}:\mathbb{R}\cup\{\infty \}&\longrightarrow&\mathbb{R^{+}}\cup\{+\infty \}\\
		t&\longmapsto&	 \| \Xi\|_{\psi_{1}}.
	\end{eqnarray*}	
	It can be clearly verified that $\mathcal{F}$ is a continuous mapping by the definition of sub-exponential norm.
	Apart from this, the mapping
	\begin{eqnarray*}
		\mathcal{G}:S^{1}&\longrightarrow&\mathbb{R}\cup\{\infty \}\\
		e^{i\theta}&\longmapsto&	 \frac{1+\tan(\frac{\theta}{2})}{1-\tan(\frac{\theta}{2})}:= t
	\end{eqnarray*}		
	 establishes the homeomorphism between $S^{1}$ and $\mathbb{R}\cup\{\infty \}$, which means $S^{1}\cong \mathbb{R}\cup\{\infty \}$.
	Now, we consider the composite mapping
	\begin{eqnarray*}
		\mathcal{F}\circ\mathcal{G}:S^{1}&\longrightarrow&\mathbb{R^{+}}\cup\{+\infty \}\\
		e^{i\theta}&\longmapsto&	 \| \Xi\|_{\psi_{1}}.
	\end{eqnarray*}
	Check carefully that $\mathcal{G}^{-1}(\mathbb{R^{+}}\cup\{+\infty \})=[-\frac{\pi}{2},\frac{\pi}{2}]$,
	thus $\mathcal{F}\circ\mathcal{G}$ is a continuous mapping on compact set $[-\frac{\pi}{2},\frac{\pi}{2}]$.
	Thereby $\mathcal{F}\circ\mathcal{G}$ is a bounded mapping since continuous mapping on compact sets must be bounded, which drives that there exists
	$ K_{0} >0$, such that $\| \Xi\|_{\psi_{1}} \leq K_{0}$ uniformly. 
	
		{\bf Step 3: Concentration.}
		Let	$l\textgreater s\in (0,1)$ be a fixed constant and $t = F_{\rho}^{-1}(1-l)$. We consider the sampling set $\Gamma=\{\Xi_1,\cdots,\Xi_m\}$ and get
		\begin{eqnarray*}
	\mu :=\sqrt{\frac{2}{1+\rho^{2}}}\int_{0 }^{t}zf_{\rho }(z)dz=\sqrt{\frac{2}{1+\rho^{2}}}H_{1}(\rho,l).
	\end{eqnarray*}
From {\bf Step 2}, we know that $\Xi$ is sub-exponential, thus by Bernstein Inequality in Lemma \ref{l_20}, we drive
	\begin{eqnarray}\label{l_22}
		\mathbb{P}\left(\left\vert {1\over m} \sum_{i=1}^m \Xi_i - \mu \right\vert>\epsilon_1\right) \leq 2e^{-c_{0}m\epsilon_1^2/K_{0}^{2}} ,
	\end{eqnarray} 
	where $K_{0}$ is the consistent sub-exponential norm bound defined in  {\bf Step 2}.
	
	Based on Lemma \ref{l_21}, by setting $\epsilon_2 = l-s \in [0, 1]$and $\eta=1-s \in [0, 1]$, we have
	\begin{eqnarray*}
		t = F_{\rho}^{-1}(1-l) \leq \widehat{F_{\rho}}^{-1}(1-s).
	\end{eqnarray*} 
	Since $\widehat{F_{\rho}}(z)$ and $\widehat{F_{\rho}}^{-1}(z)$ are monotonically increasing function, we can get
	\begin{eqnarray*}
		\widehat{F_{\rho}}(t) \leq 1-s
	\end{eqnarray*} 
	with probability  at least $1 - 4e^{-2m(l-s)^2}$. 
	Thus, we can know that at most $(1-s)$-fraction of the samples lie in $[0, t]$ with probability  at least $1 - 4e^{-2m(l-s)^2}$, and the samples lying in $[0, t]$ are smaller than those of the remaining samples.
Thus we have 
	 \begin{eqnarray}\label{l_23}
	 	\frac{1}{m}\min_{\mS \subset[m],  |\mS|\leq s m}\frac{\Vert \mathcal{A}_{\mS^{c}}(\pmb{X})\Vert_1 }{\Vert \pmb{X}\Vert_F}\ge {1\over m}\sum_{i=1}^m \Xi_i
		\end{eqnarray} 
	with probability  at least $1 - 4e^{-2m(l-s)^2}$ due to the left side of the inequality represents the smallest $(1-s)m$ samples in the sample set $\Gamma$.
	Combining (\ref{l_22}) with (\ref{l_23}), we get
	\begin{eqnarray*}
		\frac{1}{m}\min_{\mS \subset[m],  |\mS|\leq s m}\frac{\Vert \mathcal{A}_{\mS^{c}}(\pmb{X})\Vert_1 }{\Vert \pmb{X}\Vert_F}\ge \frac{1}{m}\sum_{i=1}^m  \Xi_i \ge \mu - \epsilon_1 = \sqrt{\frac{2}{1+\rho^2}}H_{1}(\rho,s+\epsilon_2)-\epsilon_1
	\end{eqnarray*}
	with certain probability. 
	By setting $\epsilon_1= \frac{l_{0}\gamma}{8}$ and $\epsilon_2 =\gamma$, we finally get (\ref{l_X1}) with probability  at least $1 -\mathcal{O}(e^{-c_{0}m\gamma^{2}})$.
	
Similarly, we can establish (\ref{l_X2}) if we 
	consider $g(x) = x\cdot \mathbbm{1}_{[t,+\infty)}(x)$ and $\Xi =g(\xi)$.		
	\end{proof}
		
\subsection{Proof of Theorem \ref{l_17} and Corollary \ref{co1}}	 
\begin{proof}[Proof of Theorem \ref{l_17}]	
By Lemma \ref{X} we can obtain that for  fixed $\pmb{X}_{0}=\pmb{x}_{0} \pmb{y}_{0}^{\tp} + \pmb{y}_{0} \pmb{x}_{0}^{\tp}\in\widetilde{T} =\{\pmb{X} = \pmb{x} \pmb{y}^{\tp} + \pmb{y} \pmb{x}^{\tp}:\pmb{x}, \pmb{y} \in \mathcal{S}^{n-1}\}$ with $\langle\pmb{x}_{0}, \pmb{y}_{0}\rangle=\rho_{0}$,
\begin{eqnarray}
		\frac{1}{m}	 \min_{\mS \subset[m], |\mS|\leq s m}\frac{\Vert \mathcal{A}_{\mS^{c}}(\pmb{X}_{0})\Vert_1 -	 \Vert \mathcal{A}_{\mS}(\pmb{X}_{0})\Vert_1}{\Vert \pmb{X}_{0}\Vert_F}  \geq\mathcal{H}^{*}(s+\gamma)-\frac{l_{0}\gamma}{4} \label{prf_thm4_1}
\end{eqnarray}
with probability at least $1 - \mathcal{O}(e^{- c m\gamma^2})$.

To complete the argument, let $\mathcal{S}_{\epsilon}$ be an $\epsilon$-net of the unit sphere, and set
\begin{eqnarray*}
	\mathcal{N}_{\epsilon}=\{\pmb{X}=\pmb{x} \pmb{y}^{\tp} + \pmb{y} \pmb{x}^{\tp}:(\pmb{x},\pmb{y})\in\mathcal{S}_{\epsilon}\times
	\mathcal{S}_{\epsilon}\}.
\end{eqnarray*}
Since $\vert\mathcal{S}_{\epsilon}\vert\leq (1+\frac{2}{\epsilon})^{n}\leq(\frac{3}{\epsilon})^{n}$, we have
\begin{eqnarray*}
	\vert\mathcal{N}_{\epsilon}\vert\leq\left(\frac{3}{\epsilon}\right)^{2n}.
\end{eqnarray*}
For any $\pmb{X}=\pmb{x} \pmb{y}^{\tp} + \pmb{y} \pmb{x}^{\tp}\in\widetilde{T}$ with $\|\pmb{x}-\pmb{x}_{0}\|_{2},\|\pmb{y}-\pmb{y}_{0}\|_{2}\le\epsilon$, we first prove that
\begin{eqnarray}
	\left\|\frac{\pmb{X}}{\|\pmb{X}\|_{F}}-\frac{\pmb{X}_{0}}{\|\pmb{X}_{0}\|_{F}}\right\|_{1}\le 12\epsilon. \label{prf_thm4_2}
\end{eqnarray}
Let $\rho=\langle\pmb{x}, \pmb{y}\rangle$, thus $\|\pmb{X}\|_{F}=\sqrt{2(1+\rho^{2})}$ and $\|\pmb{X}_{0}\|_{F}=\sqrt{2(1+\rho_{0}^{2})}$. Moreover, we get
\begin{eqnarray*}
|\rho-\rho_{0}|&=&|\langle\pmb{x}, \pmb{y}\rangle-\langle\pmb{x}_{0}, \pmb{y}_{0}\rangle|
               =|\langle\pmb{x}-\pmb{x}_{0}, \pmb{y}\rangle+\langle\pmb{x}_{0}, \pmb{y}-\pmb{y}_{0}\rangle|\\
              &\le&\|\pmb{x}-\pmb{x}_{0}\|_{2}\| \pmb{y}\|_{2}+\|\pmb{x}_{0}\|_{2}\|\pmb{y}-\pmb{y}_{0}\|_{2}\le2\epsilon,
\end{eqnarray*}
and 
\begin{eqnarray*}
\|\pmb{X}_{0}\|_{1}\le\sqrt{2}\|\pmb{X}_{0}\|_{F}=2\sqrt{1+\rho_{0}^{2}}\le2\sqrt{2}.
\end{eqnarray*}
By the triangle inequality of Frobenius norm, we also have
\begin{eqnarray*}
	\|\pmb{X}-\pmb{X}_{0}\|_{1}&\le&2\|\pmb{X}-\pmb{X}_{0}\|_{F}\\
&=&2\|(\pmb{x}-\pmb{x}_{0})\pmb{y}^{\tp}+(\pmb{y}-\pmb{y}_{0})\pmb{x}^{\tp}+
\pmb{x}_{0}(\pmb{y}-\pmb{y}_{0})^{\tp}+\pmb{y}_{0}(\pmb{x}-\pmb{x}_{0})^{\tp}\|_{F}\\
&\le&8\epsilon,
\end{eqnarray*}
where the first inequality can be drived by $\text{rank}(\pmb{X}-\pmb{X}_{0})\le \text{rank}(\pmb{X})+\text{rank}(\pmb{X}_{0})\le4$.
Then we can get
\begin{eqnarray*}
	\left\| \frac{\pmb{X}}{\|\pmb{X}\|_{F}}-\frac{\pmb{X}_{0}}{\|\pmb{X}_{0}\|_{F}} \right\|_{1}&=&
	 \frac{\left\| \|\pmb{X}_{0}\|_{F}\pmb{X}-\| \pmb{X}\|_{F}\pmb{X}_{0} \right\|_{1}}{\|\pmb{X}_0\|_{F}\|\pmb{X}\|_{F}}
	\overset{(i)}{\leq} \frac{\sqrt{2}}{2}\left\| \sqrt{1+\rho_{0}^{2}}\pmb{X} -\sqrt{1+\rho^{2}}\pmb{X}_{0}\right\|_{1}\\
	&= &\frac{\sqrt{2}}{2}\left\| \sqrt{1+\rho_{0}^{2}}(\pmb{X}-\pmb{X}_{0}) +\left(\sqrt{1+\rho_{0}^{2}}-\sqrt{1+\rho^{2}}\right)\pmb{X_{0}}\right\|_{1}\\
	&\leq& \|\pmb{X}-\pmb{X}_{0}\|_{1}+\frac{|\rho^{2}-\rho_{0}^{2}|}{\sqrt{1+\rho_{0}^{2}}+\sqrt{1+\rho^{2}}}\frac{\|\pmb{X}_{0}\|_{1}}{\sqrt{2}}\\
	&\le&8\epsilon+|\rho^{2}-\rho_{0}^{2}|
	=8\epsilon+|(\rho-\rho_{0})(\rho+\rho_{0})|\le 12\epsilon,
\end{eqnarray*}
where the inequality (i) holds since $\|\pmb{X}\|_{F}, \|\pmb{X}_0\|_{F}\ge\sqrt{2} $.
We now consider the intersection of two events
\begin{eqnarray*}
E_{1}:=\left\{	\frac{1}{m} \|\mathcal{A}(\pmb{X})\|_1 \leq  (1+\delta ) \|\pmb{X}\|_1 \text{ for all symmetric $\pmb{X}$}\right \}
\end{eqnarray*}
and
\begin{eqnarray*}
	E_2:=\left\{	\	\frac{1}{m}		\min_{\mS \subset[m],  |\mS|\leq s m} \frac{\Vert \mathcal{A}_{\mS^{c}}(\pmb{X}_{0})\Vert_1 -	 \Vert \mathcal{A}_\mS(\pmb{X}_{0})\Vert_1}{\Vert \pmb{X}_{0}\Vert_F}\geq \mathcal{H}^{*}(s+\gamma)-\frac{l_{0}\gamma}{4} \text{ for all
} \pmb{X}_{0} \in \mathcal{N}_\epsilon \right\}.
\end{eqnarray*}
Let $\mS_{1}$ be the index set of 
the largest $s$-fraction of $\mS_{\pmb{X}}=\{\langle \pmb{A}_i,\pmb{X}\rangle | i\in[m] \}$ in absolute value. 
Similarly, $\mS_{2}$ be the subset of $\mS_{\pmb{X}_{0}}$,  which collects the index of the largest $s$-fraction of $\mS_{\pmb{X}_{0}}$ in absolute value.
 Then based on the triangle inequality, we have
\begin{eqnarray*}
&&\frac{1}{m}\min_{\mS \subset[m],  |\mS|\leq s m} \frac{\Vert \mathcal{A}_{\mS^{c}}(\pmb{X})\Vert_1}{{\|\pmb{X}\|_{F}}}
=\frac{1}{m}\frac{\Vert \mathcal{A}_{\mS_{1}^{c}}(\pmb{X})\Vert_1 }{{\|\pmb{X}\|_{F}}}\\
&\ge&\frac{1}{m}\left[\left\Vert \mathcal{A}_{\mS_{1}^{c}}\left(\frac{\pmb{X}_{0}}{\|\pmb{X}_{0}\|_{F}}\right)\right\Vert_1 -\left\Vert \mathcal{A}_{S_{1}^{c}}\left(\frac{\pmb{X}}{\|\pmb{X}\|_{F}}-\frac{\pmb{X}_{0}}{\|\pmb{X}_{0}\|_{F}}\right)\right\Vert_1	\right]\\
&\ge&\frac{1}{m}\left[\left\Vert \mathcal{A}_{\mS_{2}^{c}}\left(\frac{\pmb{X}_{0}}{\|\pmb{X}_{0}\|_{F}}\right)\right\Vert_1 -\left\Vert \mathcal{A}_{\mS_{1}^{c}}\left(\frac{\pmb{X}}{\|\pmb{X}\|_{F}}-\frac{\pmb{X}_{0}}{\|\pmb{X}_{0}\|_{F}}\right)\right\Vert_1	\right].\\
\end{eqnarray*}
Similarly, we can get
\begin{eqnarray*}
	&&\frac{1}{m}\max_{\mS \subset[m],  |\mS|\leq s m} \frac{\Vert \mathcal{A}_{\mS}(\pmb{X})\Vert_1}{{\|\pmb{X}\|_{F}}}
	=\frac{1}{m}\frac{\Vert \mathcal{A}_{\mS_{1}}(\pmb{X})\Vert_1 }{{\|\pmb{X}\|_{F}}}\\
	&\le&\frac{1}{m}\left[\left\Vert \mathcal{A}_{\mS_{1}}\left(\frac{\pmb{X}_{0}}{\|\pmb{X}_{0}\|_{F}}\right)\right\Vert_1 +\left\Vert \mathcal{A}_{\mS_{1}}\left(\frac{\pmb{X}}{\|\pmb{X}\|_{F}}-\frac{\pmb{X}_{0}}{\|\pmb{X}_{0}\|_{F}}\right)\right\Vert_1	\right]\\
	&\le&\frac{1}{m}\left[\left\Vert \mathcal{A}_{\mS_{2}}\left(\frac{\pmb{X}_{0}}{\|\pmb{X}_{0}\|_{F}}\right)\right\Vert_1 +\left\Vert \mathcal{A}_{\mS_{1}}\left(\frac{\pmb{X}}{\|\pmb{X}\|_{F}}-\frac{\pmb{X}_{0}}{\|\pmb{X}_{0}\|_{F}}\right)\right\Vert_1	\right].\\
\end{eqnarray*}
Combining with the above two inequalities, we have
\begin{eqnarray*}
&&\frac{1}{m}\min_{\mS \subset[m],  |\mS|\leq s m} \frac{\Vert \mathcal{A}_{\mS^{c}}(\pmb{X})\Vert_1 -\Vert \mathcal{A}_\mS(\pmb{X})\Vert_1}{{\|\pmb{X}\|_{F}}}\\
&\ge&\frac{1}{m}\frac{\Vert \mathcal{A}_{\mS_{2}^{c}}(\pmb{X}_{0})\Vert_1 -\Vert \mathcal{A}_{\mS_{2}}(\pmb{X}_{0})\Vert_1}{\|\pmb{X}_{0}\|_{F}}	-
\frac{1}{m}\left\Vert\mathcal{A}\left(\frac{\pmb{X}}{\|\pmb{X}\|_{F}}-\frac{\pmb{X}_{0}}{\|\pmb{X}_{0}\|_{F}}\right)\right\Vert_1\\
&=&\frac{1}{m}\min_{\mS \subset[m],|\mS|\leq s m} \frac{\Vert \mathcal{A}_{\mS^{c}}(\pmb{X}_{0})\Vert_1 -\Vert \mathcal{A}_\mS(\pmb{X}_{0})\Vert_1}{{\|\pmb{X}_{0}\|_{F}}} 
- \frac{1}{m}\left\Vert\mathcal{A}\left(\frac{\pmb{X}}{\|\pmb{X}\|_{F}}-\frac{\pmb{X}_{0}}{\|\pmb{X}_{0}\|_{F}}\right)\right\Vert_1.
\end{eqnarray*}
Lemma \ref{RIP_1} yields that $\frac{1}{m} \|\mathcal{A}(\pmb{X})\|_1 \leq  (1+\delta ) \|\pmb{X}\|_1$ holds for any $\delta\in (0,\frac{1}{2})$ and all symmetric matrices $\pmb{X}$ with probability at least $1-2e^{-m\delta/8}$ when $m\ge 20\delta^{-2}n$.

Finally, by setting $\delta=\frac{l_{0}\gamma}{2}\leq 1$, $\epsilon=\frac{l_{0}\gamma}{96}$ and combining \eqref{prf_thm4_1}, \eqref{prf_thm4_2}, we can obtain
\begin{eqnarray*}
&&\frac{1}{m}\min_{\mS \subset[m],  |\mS|\leq s m} \frac{\Vert \mathcal{A}_{\mS^{c}}(\pmb{X})\Vert_1 -\Vert \mathcal{A}_\mS(\pmb{X})\Vert_1}{{\|\pmb{X}\|_{F}}}\\
&\ge&\mathcal{H}^{*}(s+\gamma)-\frac{l_{0}\gamma}{4}- (1+\delta)\left\|\frac{\pmb{X}}{\|\pmb{X}\|_{F}}-\frac{\pmb{X}_{0}}{\|\pmb{X}_{0}\|_{F}}\right\|_1
\ge\mathcal{H}^{*}(s+\gamma)-\frac{l_{0}\gamma}{4}-12(1+\delta)\epsilon\\
&\ge&\mathcal{H}^{*}(s+\gamma)-\frac{l_{0}\gamma}{2}.
\end{eqnarray*}
In conclusion, $E_{1}$ holds with probability at least $1-2e^{-c_{1}m\gamma}\geq1-2e^{-c_{1}m\gamma^{2}}$, provided $m\ge C_{2}\gamma^{-2}n$. Further, $E_{2}$ holds with probability at least
 \begin{eqnarray*}
 	1-\left(\frac{3}{\epsilon}\right)^{2n}\mathcal{O}(e^{-cm\gamma^{2}})=1-\mathcal{O}(e^{2n\ln\frac{288}{l_0\gamma}-cm\gamma^{2}})
 	\ge1-\mathcal{O}(e^{-c_{2}m\gamma^{2}}),
 \end{eqnarray*}
if $m\ge C_{2}[\gamma^{-2}\ln(288/l_0\gamma)] n$.
This concludes the proof with probability 
 \begin{eqnarray*}
\mathbb{P}(E_{1}\cap E_{2})\ge1-\mathcal{O}(e^{-c_{0}m\gamma^{2}})
\end{eqnarray*}
and $m\ge C_{0}[\gamma^{-2}\ln\gamma^{-1}] n$.
\end{proof}

\begin{proof}[Proof of Corollary \ref{co1}]
Now we set $\mS=\emptyset$ and by Proposition \ref{H}.(c), $\mathcal{H}^{*}(0)-\mathcal{H}^{*}(\tilde{\gamma})\le L_{0}\tilde{\gamma} $, then we can get   for all $\pmb{X}\in T$ with probability $1-\mathcal{O}(e^{-c_{0}m\tilde{\gamma}^{2}})$, 
 \begin{eqnarray*}	
 	\frac{1}{m}	\Vert \mathcal{A}(\pmb{X})\Vert_1 \geq \left[\mathcal{H}^{*}(0)-\left(L_{0}+\frac{l_{0}}{2}\right)\tilde{\gamma}\right]\cdot\Vert \pmb{X}\Vert_F,
 \end{eqnarray*}
provided $m \geq C_{0}[\tilde{\gamma}^{-2}\ln\tilde{\gamma}^{-1}]n$. By Proposition \ref{H}.(a),  $\mathcal{H}^{*}(0)=\frac{2\sqrt{2}}{\pi}$ and let $\gamma=\frac{\pi}{2\sqrt{2}}(L_{0}+\frac{l_{0}}{2})\tilde{\gamma}$ we can get the result in Corollary \ref{co1}.
\end{proof}

\section{ Dual Certificates}\label{dualcertificate}
Similar to \cite{candes2013phaselift,candes2014solving,candes2015phase,gross2015partial,gross2017improved}, our proof rests on the construction of dual certificates and we will  search for an allied inexact dual certificate, that is $\pmb{Y}_{T^\perp}\succeq \lambda_{0} \pmb{I}_{T^\perp}, \Vert \pmb{Y}_{T}\Vert_{F}\leq \Delta$ and its coefficients
$\pmb{y}\in \alpha\partial\|\cdot\|_{1}(-\pmb{z})$. We require $\Delta$ to be small enough then in the following (\ref{omega}), our argument can approximate the upper bound $s^{*}$. The coefficient condition due to the first order optimality conditions of (\ref{l1_min}) requires $\pmb{y}$ to belong to the subgradient of the $\ell_1$-norm.

We first construct a dual certificate for any $\pmb{x}_0$. 
Our construction mainly follows the ideas of  \cite{candes2014solving,hand2017phaselift,li2016low} and we have made some simplifications and improvements on the basis of the original structure.
Some auxiliary hypotheses are given first. Without loss of generality, we assume $\|\pmb{x}_{0}\|_{2}=1$ . Let $X \sim \mathcal{N}(0,1)$, then let $\alpha_0:= \mathbb{E} X^2 \mathbbm{ 1}(|X| \leq 3) \approx 0.9709, \,\beta_0:= \mathbb{E} X^4 \mathbbm{1}(|X| \leq 3) \approx 2.6728,\, \eta_{0}:= \mathbb{E} X^6 \mathbbm{1}(|X| \leq 3) \approx 11.2102 $. 
Let $\varepsilon $ obey Rademacher distribution as $\mathbb{\mathbb{P}}(\varepsilon=1)=\mathbb{\mathbb{P}}(\varepsilon=-1)=\frac{1}{2}$ and assume that $\varepsilon$ is independent of other random elements.
Specifically, our construction is :
\begin{eqnarray*}
	\pmb{Y} =\sum_{i\in [m]}y_{i}\pmb{a}_{i}\pmb{a}_{i}^{\tp}
	&:=&\frac{1}{m} \left[ \sum_{i \in \mS^{c}} [\beta_0 - |\langle \pmb{a}_i, \pmb{x}_0\rangle|^2 \mathbbm{1}(|\langle \pmb{a}_i,\pmb{x}_0 \rangle | \leq 3 )] \pmb{a}_{i}\pmb{a}_{i}^{\tp}
	+\sum_{i \in \mS}(9-\beta_{0})\varepsilon_{i}\pmb{a}_{i}\pmb{a}_{i}^{\tp}\right] \\
	&:=&\pmb{Y}^{(0)}-\pmb{Y}^{(1)}+\pmb{Y}^{(2)}.
\end{eqnarray*}
It is clear that $\|\pmb{y}\|_{\infty}\leq \frac{9-\beta_{0}}{m}$. Note that the existence of $\varepsilon$ makes $\mathbb{E}\left[\pmb{Y}^{(2)}\right]=\pmb{0}$. This condition can be weakened to that $\varepsilon_{i}$ and $\pmb{a}_{i}\pmb{a}_{i}^{\tp}$ are uncorrelated, that is $\mathbb{E}\left[\varepsilon_{i}\pmb{a}_{i}\pmb{a}_{i}^{\tp}\right]=0$. If this assumption is removed, we need to restrict $|\mS|=\mathcal{O}\left(n\right)$ in the analysis of $\Vert \pmb{Y}^{(2)}\Vert $ by using random matrix theory, see more details in {\bf Part 3} of Section \ref{4_1}. We hope that we can construct a new dual certificate to remove this assumption but still meet the conditions above or we can use new methods to bypass the dual certificate but still get the threshold $s^*$ in the future work.

We simply account for the motivation of the above structure. Due to the rotation invariance of the Gaussian vector, we let $\pmb{x}_{0}=\pmb{e}_{1}$. We denote the support of $T_{\pmb{x}_0}$ to be the first row and column of the matrix and the support of $T_{\pmb{x}_0}^{\perp}$ to be the remaining part. For simplification, we use $T$ and $T^{\perp}$ if there is no ambiguity.
By the definition of $T$, for matrix $\pmb{X}$, we can get
\begin{eqnarray}\label{fenjie}
	\begin{cases}
		\pmb{X}_{T}=( \pmb{x}_0 \pmb{x}_0^{\tp})\pmb{X}+(\pmb{I}-\pmb{x}_0 \pmb{x}_0^{\tp})\pmb{X}(\pmb{x}_0 \pmb{x}_0^{\tp}),\\
		\pmb{X}_{T^\perp}=(\pmb{I}-\pmb{x}_0 \pmb{x}_0^{\tp})\pmb{X}(\pmb{I}-\pmb{x}_0 \pmb{x}_0^{\tp}). 
	\end{cases}
\end{eqnarray}	
This implies that	
\begin{equation}\label{pinching}
	\| \pmb{X}_{T}\|\leq 2\|\pmb{X}\|\quad\text{and} \quad \|\pmb{X}_{T^\perp}\|\leq \|\pmb{X}\|.
\end{equation}
Thus, we can get
\begin{equation}\label{tan}
	\mathbb{E}\left[\pmb{Y}\right]=\frac{|\mS^{c}|}{m}
	\begin{bmatrix} 0 & \pmb{0} \\ 
		\pmb{0} & (\beta_0-\alpha_0) \pmb{I}_{n-1}
	\end{bmatrix}\succeq  (1-s^*) (\beta_0-\alpha_0) \pmb{I}_{T^{\perp}}   
\end{equation}
is a perfect dual certificate. 

We will show that for a fixed support set $\mS$ of adversarial sparse outlier $\pmb{z}$, dual certificates exist with high probability for any $s<s^*$. 
The concentration inequalities and the theory of non-asymptotic random matrices will be central to our proof. Our arguments are mainly derived from \cite{candes2014solving}, similar proofs that require sufficiently small bounds for $\Vert \pmb{Y}_{T}\Vert_{F}\leq \Delta$ can be found in \cite{krahmer2017phase} about sub-Gaussian sampling.

\begin{lemma} \label{dual1}
Fix the support set $\mS$ of adversarial sparse outlier $\pmb{z}$ with the signs of its nonzero entries generated from the Rademacher distribution. Set $s\textless s^{*}\approx0.1185$ and $|\mS|\le sm$. For constant $0 < \lambda ,\Delta\le1$, suppose that the number of measurements obeys $m \geq C\Big[\max(\frac{1}{\lambda^{2}},\frac{1}{\Delta})\ln\frac{1}{\min(\lambda,\Delta)}\Big]n$, where $C$ is a sufficiently large numerical constant, then for all $\pmb{x}\in \mathbb{R}^{n}$ there exists $\pmb{Y}$ in the range of $\mathcal{A}^{*}$ obeying:

	\begin{equation}\label{l_9}
		\pmb{Y}_{T^\perp}\succeq \Big[(1-s^*) (\beta_0-\alpha_0)-\lambda\Big] \pmb{I}_{T^\perp} ,\quad \Vert \pmb{Y}_{T}\Vert_{F}\leq \Delta,
	\end{equation}
	and
	\begin{equation}\label{l_4}
		\begin{cases}
			y_{i}=-\frac{9-\beta_{0}}{m} \text{sgn}(z_{i})& z_{i}\in \mS\\
			\lv y_{i}\rv\leq\frac{9-\beta_{0}}{m} & z_{i}\in \mS^{c} 
		\end{cases}
	\end{equation}
	with probability at least $1-6e^{-c_{0}m\lambda^{2}}-5e^{-c_{0}m\Delta }$, where $c_{0}$ is a positive absolute constant.
\end{lemma}

\subsection{ Positive Definiteness of $\pmb{Y}_{T^{\perp}}$}\label{4_1}
We first check that when $\pmb{x}_{0}$ and the support of $\pmb{z}$ are fixed, with probability at least $1-6e^{-cm\lambda^{2}}$, the condition $$\pmb{Y}_{T^{\perp}}\succeq\left[(\beta_{0}-\alpha_{0})(1-s^{*})-\frac{\lambda}{2}\right]\pmb{I}_{T^{\perp}}$$ is available, provided $m\ge C\lambda^{-2}n$. 

{\bf Part 1:} Firstly, 
	it is clear that $\frac{m}{|\mS^{c}|}\pmb{Y}^{(0)}$ is a Wishart matrix as
	\begin{eqnarray*}
		\frac{m}{|\mS^{c}|}\pmb{Y}^{(0)} =\frac{1}{|\mS^{c}|}\sum_{i\in \mS^{c}}\beta_{0}\pmb{a}_{i}\pmb{a}_{i}^{\tp} = \frac{\beta_0}{|\mS^{c}|}\widetilde{\pmb{A}}^{\tp}\widetilde{\pmb{A}},
	\end{eqnarray*}
where $\widetilde{\pmb{A}} \in \mathbb{R}^{|\mS^{c}|\times n}$ is the matrix composed of row vectors $\{ \pmb{a}_i^{\tp}\}_{i\in \mS^c}$.
	Thus, by Lemma \ref{random} we have
	\begin{eqnarray*}
		\left\lVert \frac{m}{|\mS^{c}|}\pmb{Y}^{(0)}-\beta_{0}\pmb{I}\right\rVert = \beta_0\left\Vert \frac{1}{|\mS^c|}\widetilde{\pmb{A}}^{\tp}\widetilde{\pmb{A}} -\pmb{I}\right\Vert \leq \beta_0\max(\delta,\delta^2)
	\end{eqnarray*}
	with probability at least $1 - 2e^{-t^2/2}$, where $\delta =\sqrt{\frac{n}{|\mS^c|}} + \frac{t}{\sqrt{|\mS^c|}}$.
	By setting $t = \sqrt{|\mS^c|}\lambda$ and $|\mS^{c}| \geq (1-s)m \ge C_{1}\lambda^{-2}n$, we have $\delta=\mathcal{O}(\lambda)$. Then by choosing suitable constant $\lambda$, we can get 
	\bea
	\left\lVert \frac{m}{|\mS^{c}|}\pmb{Y}^{(0)}-\beta_{0}\pmb{I}\right\rVert \le\frac{\lambda}{6}
		\le\frac{\lambda m}{6|\mS^{c}|}
	\eea
	with probability at least $1-2e^{-|\mS^{c}|\lambda^{2}/2}$. 
	It should be noted again that $|\mS^{c}|\ge(1-s)m$.
	Specifically, by (\ref{pinching}) with probability at least $1-2e^{-(1-s^{*})m\lambda^{2}/2}$, 
we have
	\begin{eqnarray}\label{a_1} 
		\left\lVert \pmb{Y}^{(0)}_{T^\perp}-\beta_{0}\frac{|\mS^{c}|}{m}\pmb{I}_{T^{\perp}}\right\rVert\le\frac{\lambda}{6},
	\end{eqnarray}
provided $m\ge \frac{C_{1}}{1-s^{*}}\lambda^{-2}n$. 
	
	{\bf Part 2:} Secondly, let $\pmb{\tilde{a}}$ be the projection of $\pmb{a}$ onto the orthogonal 
	complement of span($\pmb{x}_{0}$), that is $\pmb{\tilde{a}}=(\pmb{I}-\pmb{x}_{0}\pmb{x}_{0}^{\tp})\pmb{a}$. 
	From the conclusions in \cite{candes2014solving}, we can get $\mathbb{E}\pmb{\zeta}_{i}\pmb{\zeta}_{i}^{T}=\alpha_{0}\pmb{I}_{T^{\perp}}$, since $\pmb{\zeta}_{i}=\langle \pmb{a}_i, \pmb{x}_0\rangle \mathbbm{1}(|\langle \pmb{a}_i,\pmb{x}_0 \rangle |\leq 3 )\pmb{\tilde{a}_{i}}$  are zero-mean, isotropic and sub-Gaussian random vectors.  
	Therefore, 
	similar to {\bf Part 1}, Lemma \ref{random} asserts that with probability at least $1-2e^{-c_2m\lambda^{2}}$,
	\begin{eqnarray}\label{a_2}
		\left\lVert \pmb{Y}^{(1)}_{T^\perp}-\alpha_{0}\frac{|\mS^{c}|}{m}\pmb{I}_{T^{\perp}}\right\rVert\le\frac{\lambda}{6},
	\end{eqnarray}	
	provided $m\ge C_2\lambda^{-2}n$, where $C_2=C_{K},c_2=c_{K}$ are related to $K=\|\pmb{\zeta}_{i}\|_{\psi_{2}}$.
	
	{\bf Part 3:} Thirdly, we provide the following lemma for dealing with the spectral norm of $\pmb{Y}^{(2)}$, actually it 
	can be traced back to the Rademacher process in \cite{koltchinskii2015bounding, mendelson2014learning,tropp2015convex,kueng2017low,krahmer2020complex}. The specific certification details will be placed in \ref{aaa}.
	\begin{lemma}\label{a_4}
		Set $\pmb{\Phi }=\frac{1}{N}\sum\limits_{i\in [N]}\varepsilon _{i}\pmb{a}_{i}\pmb{a}_{i}^{\tp}$, where $\{\pmb{a}_i \}_i$ are standard Gaussian vectors. Then for every $ t\ge0 $, there exists numerical constant $ c>0$, such that 
		\begin{eqnarray*}
			\mathbb{P}(	\|\pmb{\Phi }\|\ge t)\le2e^{-cN\min\left(\frac{t^{2}}{4K_0^2},\frac{t}{2K_0}\right)},
		\end{eqnarray*}
	provided $N\ge C\min\left(\frac{t^{2}}{4K_0^2},\frac{t}{2K_0}\right)^{-1}n$. Here, $K_0 = \Vert X \Vert_{\psi_{2}}^2 = \frac{8}{3}$ when $X\sim \mathcal{N}(0,1)$.
	\end{lemma}
	
		By the above lemma,
	set $t=\frac{\lambda}{6}\frac{m}{(9-\beta_{0})|\mS|}$, then with high probability, 
	we have 
	\begin{eqnarray}\label{a_3}
		\left\lVert\pmb{Y}^{(2)}_{T^\perp}\right\rVert\le \left\lVert\pmb{Y}^{(2)}\right\rVert 
		= \frac{9-\beta_0}{m}|\mS|\left\Vert \frac{1}{|\mS|}\sum_{i\in\mS} \varepsilon _{i}\pmb{a}_{i}\pmb{a}_{i}^{\tp} \right\Vert
		 \le (9-\beta_{0})t \frac{|\mS|}{m}= \frac{\lambda}{6}.
	\end{eqnarray}
	If $t\le 2K_0$, thus $ \frac{\lambda}{12(9-\beta_{0})K_0}\leq \frac{|\mS|}{m}\textless s^{*}$, (\ref{a_3}) holds 
	with probability at least $1-2e^{-c_{3}m^{2}\lambda^{2}/|\mS|}\ge1-2e^{-c_{3}m\lambda^{2}/s^{*}}$ and provided $|\mS|\ge C_{3}n|\mS|^2/(\lambda^{2}m^2)$, which means $m\ge C_{3}s^{*}\lambda^{-2}n$ is sufficient. 
	If $t\ge 2K_0$, thus $ 0\le \frac{|\mS|}{m}\textless \frac{\lambda}{12(9-\beta_{0})K_0}$, (\ref{a_3}) holds 
	with probability at least $1-2e^{-c_{4}m\lambda} \geq 1-2e^{-c_{4}m\lambda^2}$ for sufficiently small $\lambda$ and provided 
	$m\ge C_4\lambda^{-1}n$.

	{\bf Part 4:} Combine (\ref{a_1})-(\ref{a_3}), we have
	\begin{eqnarray*}
		\left\lVert\pmb{Y}_{T^\perp}-(\beta_{0}-\alpha_{0})\frac{|\mS^{c}|}{m}\pmb{I}_{T^\perp}\right\rVert\le \frac{\lambda}{2}.
	\end{eqnarray*}
	Finally we get the desired result that 
	\begin{eqnarray*}
		\pmb{Y}_{T^\perp}\succeq \left[(\beta_{0}-\alpha_{0})(1-s^{*})-\frac{\lambda}{2}\right]\pmb{I}_{T^\perp} \succ \pmb{0}
	\end{eqnarray*}
	holds with probability at least $1-6e^{-cm\lambda^{2}}$, provided $m\ge C\lambda^{-2}n$.
	
	\subsection{Boundness of $\|\pmb{Y}_{T}\|_{F}$}
	We will establish that when we fix $\pmb{x}_{0}$ and the support of $\pmb{z}$, with probability at least $1-5e^{-cm\Delta}$, the Frobenius norm of $\pmb{Y}_{T}$ can be small enough as $$\|\pmb{Y}_{T}\|_{F}\leq 2\Delta/3,$$ provided $m\ge C\Delta^{-1}n$.
	
	{\bf Part 1:} First of all, following \cite{candes2014solving,krahmer2017phase}, we can write
	\begin{eqnarray*}
		\|\pmb{Y}_{T}\|_{F}^{2}= |\pmb{x}^{\tp}_0\pmb{Y}\pmb{x}_0|^{2}+2\|(\pmb{I}-\pmb{x}_0\pmb{x}^{\tp}_0)\pmb{Y}\pmb{x}_0 \|_{2}^{2},
	\end{eqnarray*} 
	thus it suffices to prove that the following two inequalities hold with high probability:
	\begin{eqnarray*}
		|\pmb{x}^{\tp}_0\pmb{Y}\pmb{x}_0|^{2}\le \Delta/3\quad \text{and} \quad \|(\pmb{I}-\pmb{x}_0\pmb{x}^{\tp}_0)\pmb{Y}\pmb{x}_0 \|_{2}^{2}\le \Delta/6.
	\end{eqnarray*}
	
	{\bf Part 2:} Secondly, set $X_{i}=\langle \pmb{a}_{i}, \pmb{x}_0\rangle\sim\mathcal{N}(0,1)$
	and then let $\pmb{x}^{\tp}_0\pmb{Y}\pmb{x}_0=\frac{1}{m}\sum\limits_{i \in [m] }\xi_{i} $ where $\xi_{i}$ are independent
	zero-mean variables:
	\begin{eqnarray*}
		\xi_{i}=	
		\begin{cases}
			\beta_{0}X_{i}^{2}-X_{i}^{4}\mathbbm{1}(|X_{i} | \leq 3 )&i\in \mS^{c}\\
			(9-\beta_{0})\varepsilon_{i} X_{i}^{2}&i\in \mS.\\
		\end{cases}
	\end{eqnarray*}
	Note that the sub-exponential norm of $X_{i}^{4}\mathbbm{1}(|X_{i}| \leq 3 )$ is bounded with upper bound $K_{1}\le 4$ and $X_{i}^{2}$ are $\chi^{2}(1)$ random variables, then 
	\begin{eqnarray*}
		\left\lVert\beta_{0}X_{i}^{2}-X_{i}^{4}\mathbbm{1}(|X_{i} | \leq 3 )\right\rVert_{\psi_{1}}\le\beta_{0}\left\lVert X_{i}^{2}\right\rVert_{\psi_{1}}+K_{1} \leq \beta_0 K_0 + K_1
	\end{eqnarray*}
	and 
	\begin{eqnarray*}
		\left\lVert(9-\beta_{0})\varepsilon_{i} X_{i}^{2}\right\rVert_{\psi_{1}}=(9-\beta_{0})\left\lVert X_{i}\right\rVert^{2}_{\psi_{2}} =(9-\beta_{0})K_0,
	\end{eqnarray*}
	where we use the fact that $\left\lVert\varepsilon_{i} X_{i}^{2}\right\rVert_{\psi_{1}}=\left\lVert X_{i}\right\rVert^{2}_{\psi_{2}}=K_0 = \frac{8}{3}$.
	Thus, $\xi_{i}$ are mean-zero and sub-exponential variables with bounded sub-exponential norm 
	\begin{eqnarray*}
		K_{\xi}&=& \max\left(\left\lVert\beta_{0}X_{i}^{2}-X_{i}^{4}\mathbbm{1}(|X_{i} | \leq 3 )\right\rVert_{\psi_{1}},	\left\lVert(9-\beta_{0})\varepsilon_{i} X_{i}^{2}\right\rVert_{\psi_{1}} \right) <+\infty.
	\end{eqnarray*}
	Therefore, by Bernstein Inequality in Lemma \ref{l_20}, we have
	\begin{eqnarray*}
		\mathbb{P}\left(\pmb{x}^{\tp}_0\pmb{Y}\pmb{x}_0\ge \sqrt{\Delta/3}\right)\le 2e^{-\gamma m} = 2e^{-cm\Delta}, 
	\end{eqnarray*}
	where $\gamma=c_{5} \min\left(\frac{\Delta/3}{K_{\xi}^{2}},\frac{\sqrt{\Delta/3}}{K_{\xi}}\right)=\mathcal{O}(\Delta)$ 
	for sufficiently small  $\Delta<3K^{2}_{\xi}$.
	
	{\bf Part 3:} Finally, as $X_{i}=\langle \pmb{a}_{i}, \pmb{x}_0\rangle\sim\mathcal{N}(0,1)$, we let 
	\begin{eqnarray*}
		\mu _{i}=	
		\begin{cases}
			\beta_{0}X_i - X_{i}^{3}\mathbbm{1}(|X_{i}| \leq 3 )&i\in \mS^{c}\\
			(9-\beta_{0})\varepsilon_i X_{i}&i\in \mS\\
		\end{cases}
	\end{eqnarray*}
	and $\pmb{\widetilde{U}}=
	\begin{bmatrix} 
		\pmb{\tilde{a}}_{1} ,&\cdots& ,\pmb{\tilde{a}}_{m}
	\end{bmatrix}$. Similar to \cite{candes2014solving}, we can write $(\pmb{I}-\pmb{x}_0\pmb{x}^{\tp}_0)\pmb{Y}\pmb{x}_0 =\frac{1}{m}\pmb{\widetilde{U}}\pmb{\mu}$ where $\pmb{\widetilde{U}}$ and $\pmb{\mu} $ 
	are independent because $\pmb{\tilde{a}}_{i}$ are the projection of $\pmb{a}_{i}$ onto the orthogonal complement of span($\pmb{x}_{0}$). A direct calculation gives
	\begin{eqnarray*}
		\mathbb{E}\left[\mu _{i}^{2}\right]=	
		\begin{cases}
			\eta_{0}-\beta_{0}^{2}&i\in \mS^{c}\\
			(9-\beta_{0})^{2}&i\in \mS.\\
		\end{cases}
	\end{eqnarray*}
Similar to {\bf Part 2}, We can also know that $\mu_{i}^{2}-\mathbb{E}[\mu _{i}^{2}]$ are sub-exponential variables with bounded sub-exponential norm $K_{\mu}$.
		Thus, by Bernstein Inequality in Lemma \ref{l_20}, we have
	\begin{eqnarray*}
		\mathbb{P}\left(\left\lVert \pmb{\mu}\right\rVert_{2}^{2}-\mathbb{E}\left\lVert\pmb{\mu}\right\rVert_{2}^{2} \ge m\right)=	\mathbb{P}\left( \frac{\sum_{i=1}^{m}(\mu_{i}^{2}-\mathbb{E}\mu_{i}^{2})}{m}\ge 1\right)\le 2e^{-\gamma m},
	\end{eqnarray*}
	where $\gamma=\min\left(\frac{c_{6}}{K_{\mu}^{2}},\frac{c_{6}}{K_{\mu}}\right)$ and the expectation of $\left\lVert \pmb{\mu}\right\rVert_{2}^{2}$ is
	\begin{eqnarray*}
		\mathbb{E}\left\lVert \pmb{\mu}\right\rVert_{2}^{2}
		&=&|\mS^{c}|(\eta_{0}-\beta_{0}^{2})+|\mS|(9-\beta_{0})^{2}\\
		&\le&(1-s^{*})m(\eta_{0}-\beta_{0}^{2})	+s^{*}m(9-\beta_{0})^{2}:=\tau_{0}m \le 10m.
	\end{eqnarray*}	
	This gives that with probability at least $1-2e^{-\gamma m}$,
	\begin{eqnarray}\label{a}
		\left\lVert \pmb{\mu}\right\rVert_{2}^{2}\le (1+\tau_{0})m.
	\end{eqnarray}	
	Besides, for fixed $\pmb{x}\in\mS^{n-1}$, $\left\lVert \pmb{\widetilde{U}}\pmb{x}\right\rVert_{2}^{2} $
	is a $\chi^{2}(n-1)$ random variable, and Chernoff upper bound implies
	\begin{eqnarray}\label{b}
		\mathbb{P}\left(\left\lVert \pmb{\widetilde{U}}\pmb{x}\right\rVert_{2}^{2} \ge \frac{m\Delta}{66}\right)\le e^{-cm\Delta},
	\end{eqnarray}
	for some numerical constant $c \textgreater 0$ and sufficiently large constant $C$ with $m\ge C\Delta^{-1}n$.
	Now combining (\ref{a}) with (\ref{b}), we can get with probability at least $1-3e^{-cm\Delta}$,
	\begin{eqnarray*}
		\left\| \left(\pmb{I}-\pmb{x}_0\pmb{x}^{\tp}_0\right)\pmb{Y}\pmb{x}_0 \right\|_{2}^{2}=\frac{1}{m^{2}}\left\lVert \pmb{\widetilde{U}}\pmb{\mu}\right\rVert_{2}^{2}\le \frac{(1+\tau_{0})\Delta}{66} \le \Delta/6,
	\end{eqnarray*}
	provided $m\ge C\Delta^{-1}n$.

	\subsection{Continuity Argument }
	We have shown that for the fixed support set $\mS$ of adversarial sparse outlier $\pmb{z}$ and fixed $\pmb{x}$, there is a dual certificate satisfying the requirements. Then the continuity argument allows us to extend it to all $\pmb{x}\in\mathcal{S}^{n-1}$ and we can further extend  this conclusion to all $\pmb{x}\in\mathbb{R}^{n}$ by the scaling of the dual certificates, see \cite{hand2017phaselift} for details. 
	\begin{proof}[\text{Proof of Lemma \ref{dual1}}]
		Let $\mathcal{N}_{\epsilon}$ be an $\epsilon$-net for the $\mathcal{S}^{n-1}$ with cardinality $|\mathcal{N}_{\epsilon}|\le (1+\frac{2}{\epsilon})^{n}\le(\frac{3}{\epsilon})^{n}$. Thus, with probability at least  $ 1-\Big(6e^{-cm\lambda^{2} }+5e^{-cm\Delta }\Big)(\frac{3}{\epsilon})^{n},$ for all $\pmb{x}_{0}\in \mathcal{N}_{\epsilon}$, there exists $\pmb{Y}=\mathcal{A}^{*}(\pmb{y})$ satisfying (\ref{l_9}) and (\ref{l_4}).
				
		Let $\pmb{x}_{0}\in \mathcal{N}_{\epsilon}$ such that $\|\pmb{x}-\pmb{x}_{0}\|\le\epsilon$. For any $\pmb{x}\in\mS^{n-1}$, we set $\pmb{\Lambda}=\pmb{x}\pmb{x}^{\tp}-\pmb{x}_{0}\pmb{x}_{0}^{\tp}$ and note that $\|\pmb{\Lambda}\|_{F}\le 2\epsilon$. We have
		\begin{eqnarray*}	
			\pmb{Y}_{T^{\perp}}=\pmb{Y}_{T_{0}^{\perp}}-\pmb{R},\quad
			\pmb{Y}_{T}=\pmb{Y}_{T_{0}}+\pmb{R}.
		\end{eqnarray*}	
		Here, $T$ and $T_{0}$ represent the tangent spaces of $\pmb{x}$ and $\pmb{x}_{0}$ respectively and $\pmb{R}$ originates from (\ref{fenjie}), the specific expression is
		\begin{eqnarray}	
			\pmb{R}=\pmb{\Lambda}\pmb{Y}(\pmb{I}-\pmb{x}_{0}\pmb{x}_{0}^{\tp})+(\pmb{I}-\pmb{x}_{0}\pmb{x}_{0}^{\tp})\pmb{Y}\pmb{\Lambda}-\pmb{\Lambda}\pmb{Y}\pmb{\Lambda}.
		\end{eqnarray}
		Note that for $\Delta,\lambda\le1$, we have
		\begin{eqnarray*}
			\|\pmb{Y}\|\le \|\pmb{Y}_{T_{0}^{\perp}}\|+\|\pmb{Y}_{T_{0}}\|\le\left(\beta_{0}-\alpha_{0}+\frac{\lambda}{2}\right)+2\Delta/3\le3.
		\end{eqnarray*}
	Then we can get 
		\begin{eqnarray*}	
			\|	\pmb{R}\|\le2	\|\pmb{Y}\|	\|\pmb{\Lambda}\|\|\pmb{I}-\pmb{x}_{0}\pmb{x}_{0}^{\tp}\|+	\|\pmb{Y}\|\|\pmb{\Lambda}\|^{2}\le12(\epsilon+\epsilon^{2}).
		\end{eqnarray*}
		Since $\pmb{\Lambda}$ has rank at most 2, then $\text{rank}(\pmb{R})\le2$ and $\|\pmb{R}\|_{F}\le\sqrt{2}\|	\pmb{R}\|\le12\sqrt{2}(\epsilon+\epsilon^{2}).$
		Now, we can choose $\epsilon=\min(\Delta,\lambda)/80$ such that $20\sqrt{2}(\epsilon+\epsilon^{2})\le\Delta/3$ and $20(\epsilon+\epsilon^{2})\le \frac{\lambda}{2}$.
		
		Summarily, for any $s<s^*\approx0.1185$, we can construct dual certificates satisfying (\ref{l_9}) and (\ref{l_4}) for all $\pmb{x}\in\mathcal{S}^{n-1}$ with probablity at least 
		\begin{eqnarray*}
			1-\Big(6e^{-cm\lambda^{2} }+5e^{-cm\Delta }\Big)\left(\frac{3}{\epsilon}\right)^{n}&\ge&
			1-\Big(6e^{-cm\lambda^{2}+\tilde{c}\ln(\epsilon^{-1})n } + 5e^{-cm\Delta +\tilde{c}\ln(\epsilon^{-1})n}\Big)\\
			&=&	1-\Big(6e^{-c_{0}m\lambda^{2} }+5e^{-c_{0}m\Delta }\Big),
		\end{eqnarray*}
		provided $m\ge C\Big[\max(\frac{1}{\lambda^{2}},\frac{1}{\Delta})\ln\frac{1}{\min(\lambda,\Delta)}\Big]n$ for sufficiently large $C$.
	\end{proof}

\section{Proof of Main Results}\label{prfs}
\subsection{Core Lemma}
Our strategy to prove Theorem \ref{theorem} and Theorem \ref{theorem2} mainly depends on the $s$-robust outlier bound condition and the inexact dual certificate. Here we give the core lemma.
\begin{lemma}\label{l_1}
	Suppose that the linear mapping $\mathcal{A}$ obeys 
	\be
	\frac{1}{m} \|\mathcal{A}(\pmb{X})\|_1 \leq  (1+\delta ) \|\pmb{X}\|_1 \label{A2}
	\ee
	for all positive semidefinite matrices $\pmb{X}$,
	and  $s$-robust outlier bound condition (\ref{rip1}) for all matrices $\pmb{X} \in T$.
	Suppose that there exists $\pmb{Y}=\mathcal{A^{*}}\pmb{y}$ obeying (\ref{l_9}) and (\ref{l_4}).
	Set $\lambda_{0}=(1-s^{*})(\beta_{0}-\alpha_{0})-\lambda$, under the condition 
	\begin{eqnarray}\label{omega}
		\Omega:= \lambda_{0}\frac{C(s)}{1+\delta}- \Delta > 0,
	\end{eqnarray} the solution of (\ref{l1_min}) satisfies
	\begin{eqnarray*}
		\Vert\widehat{\pmb{X}} -\pmb{X}_{0}\Vert_F\leq C_{0}\frac{\Vert\pmb{\omega}\Vert_1}{m},
	\end{eqnarray*}
	for some constant 
	\begin{eqnarray}\label{C}
		C_{0}=\frac{(18-2\beta_{0}+2\lambda_{0})(1+ \Delta/\lambda_{0}  )}{\Omega}+\frac{18-2\beta_{0}}{\lambda_{0}}.
	\end{eqnarray}
\end{lemma}
Before proving the above lemma, we first show that the one-side restricted isometry property (\ref{A2}) and the $s$-robust outlier bound condition will restrict the optimal solution $\widehat{\pmb{X}}=\pmb{X}_{0}+\pmb{H}$ of (\ref{l1_min}) to a specified norm cone.
\begin{lemma}\label{l_2}
	Suppose that the linear mapping $\mathcal{A}$ obeys the   properties (\ref{A2}) and (\ref{rip1}), then any optimal solution $\widehat{\pmb{X}}=\pmb{X}_{0}+\pmb{H}$ satisfies
	\begin{equation}
		\frac{C(s)}{1+\delta} \cdot\Vert \pmb{H}_{T}\Vert_F\leq \frac{2}{m}\|\pmb{\omega}  \|_1+\|\pmb{H}_{T^\perp}\|_1.
	\end{equation}	
\end{lemma}
\begin{proof}
	Let $\widehat{\pmb{X}}=\pmb{X}_{0}+\pmb{H}$ be a minimizer of (\ref{l1_min}), then we have $\widehat{\pmb{X}}\succeq \pmb{0}$ and $\pmb{H}_{T^\perp}\succeq \pmb{0} $.	
	By assumption,
	\begin{eqnarray*}
		\| \mathcal{A}(\pmb{X}_{0}+\pmb{H})  - \pmb{b}  \|_1 \leq \|\mathcal{A}(\pmb{X}_{0})  - \pmb{b}  \|_1.
	\end{eqnarray*}
	As $ \pmb{b}=\mathcal{A}(\pmb{X}_{0})+\pmb{z}+\pmb{\omega }$, we have
	\begin{equation}\label{l_6}
		\| \mathcal{A}(\pmb{H})  - \pmb{z} -\pmb{\omega } \|_1 \leq \|   \pmb{z}+\pmb{\omega }  \|_1.
	\end{equation}
	Since $\ell_{1}$-norm can be divided into  parts $\mS$ and $\mS^{c}$,
	we can get 
	\begin{eqnarray*}
		\| \mathcal{A}_{\mS}(\pmb{H})  - \pmb{z} -\pmb{\omega}_{\mS} \|_1 +	\| \mathcal{A}_{{\mS}^{c}}(\pmb{H})  -\pmb{\omega}_{{\mS}^{c}} \|_1 \leq \|   \pmb{z}+\pmb{\omega}_{\mS}  \|_1+ \|\pmb{\omega}_{{\mS}^{c}}  \|_1.
	\end{eqnarray*}
	Using the triangle inequality, 
	we could further bound
	\begin{equation}
		\|  \mathcal{A}_{{\mS}^{c}}(\pmb{H})  \|_1-\| \mathcal{A}_{\mS}(\pmb{H})  \|_1 \leq  2\|\pmb{\omega}_{{\mS}^{c}}  \|_1.
	\end{equation}
	Similarly, since $\pmb{H}=\pmb{H}_{T}+\pmb{H}_{T^\perp}$, by using the triangle inequality again, we know
	\begin{eqnarray*}
		\| \mathcal{A}_{\mS}(\pmb{H})  \|_1&\leq&\| \mathcal{A}_{\mS}(\pmb{H}_{T})  \|_1  +  \| \mathcal{A}_{\mS}(\pmb{H}_{T^\perp})  \|_1,\\
		\| \mathcal{A}_{{\mS}^{c}}(\pmb{H})  \|_1 &\geq& \| \mathcal{A}_{{\mS}^{c}}(\pmb{H}_{T})  \|_1
		-\| \mathcal{A}_{{\mS}^{c}}(\pmb{H}_{T^\perp})  \|_1.
	\end{eqnarray*}
	Furthermore, we obtain 
	\begin{equation}
		\| \mathcal{A}_{{\mS}^{c}}(\pmb{H}_{T})  \|_1- \| \mathcal{A}_{\mS}(\pmb{H}_{T})  \|_1 \leq  2\|\pmb{\omega}_{{\mS}^{c}}  \|_1     +   \| \mathcal{A}(\pmb{H}_{T^\perp})  \|_1.
	\end{equation}
	
	\noindent From the $s$-robust outlier bound condition (\ref{rip1}), and the assumption on $\mathcal{A}$ of (\ref{A2}), we conclude that 
	\begin{eqnarray*}
		C(s)\cdot\Vert \pmb{H}_{T}\Vert_F\leq \frac{2}{m}\|\pmb{\omega}_{{S}^{c}}  \|_1  + (1+\delta ) \|\pmb{H}_{T^\perp}\|_1.
	\end{eqnarray*}
	Thus, the optimal solution lies on the norm cone 
	\begin{eqnarray*}
		\frac{C(s)}{1+\delta} \cdot\Vert \pmb{H}_{T}\Vert_F\leq \frac{2}{m}\|\pmb{\omega}  \|_1+\|\pmb{H}_{T^\perp}\|_1.
	\end{eqnarray*}
\end{proof}
We now use the Lemma \ref{l_2} and the inexect dual certificate to prove the core lemma, mainly using the property of subgradient and some standard descriptions.
\begin{proof}[Proof of Lemma \ref{l_1}]
	Due to the assumption of (\ref{l_4}) and let $\kappa=(9-\beta_{0})/m$, we can get 
	\begin{eqnarray*}
		\pmb{y }/\kappa\in \partial\|\cdot\|_{1}(-\pmb{z}),
	\end{eqnarray*} 	
	thus  the definition of subgradient and (\ref{l_6}) give
	\begin{eqnarray*}
		\|-\pmb{z}\|_{1}+\langle\pmb{y }/\kappa, \mathcal{A}(\pmb{H})-\pmb{\omega}\rangle\leq \|\mathcal{A}(\pmb{H})-\pmb{z}-\pmb{\omega}\|_{1}
		\leq  \|   \pmb{z}+\pmb{\omega }  \|_1
		\leq  \|   \pmb{z}  \|_1+\|\pmb{\omega }\|_1.
	\end{eqnarray*}
	Hence,
	\begin{eqnarray*}
		\langle\pmb{y}, \mathcal{A}(\pmb{H})-\pmb{\omega}\rangle\leq \kappa\|\pmb{\omega }\|_1
	\end{eqnarray*} 
	and then we can get
	\begin{eqnarray*}
		\langle\pmb{Y}, \pmb{H}\rangle = \langle\pmb{y }, \mathcal{A}(\pmb{H})\rangle
		\leq \langle\pmb{y }, \pmb{\omega}\rangle+\kappa\|\pmb{\omega }\|_1
		\leq  \|\pmb{y }\|_{\infty }\|\pmb{\omega }\|_1+\kappa\|\pmb{\omega }\|_1
		\leq  2\kappa\|\pmb{\omega }\|_1.
	\end{eqnarray*}
	Decompose the above formula into $T$ and $T^\perp$, we have
	\begin{equation*}
		\langle\pmb{Y}_{T^\perp}, \pmb{H}_{T^\perp}\rangle\leq |\langle\pmb{Y}_{T}, \pmb{H}_{T}\rangle|+2\kappa\|\pmb{\omega }\|_1.
	\end{equation*}
	Then through condition (\ref{l_9}) 
	we get
	\begin{eqnarray*}
		\lambda_{0}\| \pmb{H}_{T^\perp}\|_{1}
		&\leq &|\langle\pmb{Y}_{T^\perp}, \pmb{H}_{T^\perp}\rangle|
		\leq  |\langle\pmb{Y}_{T}, \pmb{H}_{T}\rangle|+2\kappa\|\pmb{\omega }\|_1\\
		&\leq &\|\pmb{Y}_{T}\|_{F}\| \pmb{H}_{T}\|_{F}+2\kappa\|\pmb{\omega }\|_1
		\leq \Delta \cdot  \|\pmb{H}_{T}\|_{F}+2\kappa\|\pmb{\omega }\|_1.
	\end{eqnarray*}
	Combining the above inequality with Lemma \ref{l_2}, we could bound 
	\begin{eqnarray*}
	\lambda_{0}\frac{C(s)}{1+\delta} \cdot\Vert \pmb{H}_{T}\Vert_F\leq  \Delta \cdot  \|\pmb{H}_{T}\|_{F}+\left(\frac{2\lambda_{0}}{m}+2\kappa\right)\|\pmb{\omega }\|_1.
	\end{eqnarray*}
	One more step,
	\begin{equation}
		\Omega\cdot\Vert \pmb{H}_{T}\Vert_F=\left[\lambda_{0}\frac{C(s)}{1+\delta}- \Delta \right] \cdot\Vert \pmb{H}_{T}\Vert_F\leq \left(\frac{2\lambda_{0}}{m}+2\kappa\right)\|\pmb{\omega }\|_1.
	\end{equation}
	Finally,
	\begin{eqnarray*}
		\left\|\widehat{\pmb{X}}-\pmb{X}_{0}\right\|_{F}
		&= &\|\pmb{H}\|_{F}
		\leq \|\pmb{H}_{T}\|_{F}+\|\pmb{H}_{T^\perp}\|_{1}
		\leq \left(1+ \frac{\Delta}{\lambda_{0}}\right)\cdot  \|\pmb{H}_{T}\|_{F}+2\kappa/\lambda_{0}\|\pmb{\omega }\|_1\\
		&\leq &\frac{(\frac{2\lambda_{0}}{m}+2\kappa)(1+ \Delta/\lambda_{0}  )}{\Omega}\|\pmb{\omega }\|_1 +2\kappa/\lambda_{0}\|\pmb{\omega }\|_1
		\leq \frac{C_{0}}{m}\|\pmb{\omega }\|_1,
	\end{eqnarray*}
	for constant 	
	\begin{eqnarray*}
		C_{0}&=&\frac{(18-2\beta_{0}+2\lambda_{0})(1+ \Delta/\lambda_{0}  )}{\Omega}+\frac{18-2\beta_{0}}{\lambda_{0}}.
	\end{eqnarray*}
\end{proof}
 
\subsection{Proof of Theorem \ref{theorem} and Theorem \ref{theorem2}}
We now aggregate the various intermediate results to establish main results. Since Theorem \ref{theorem} is a special case of Theorem \ref{theorem2} if $\pmb{\omega}=\pmb{0}$, we only prove Theorem \ref{theorem2}.
\begin{proof}[Proof of Theorem \ref{theorem2}]
Theorem \ref{theorem2} can be derived directly by Lemma \ref{l_1} and the proof of the error estimate $\|\hat{\pmb{x}}-(-1)^{k}\pmb{x}_{0}\|_{2}$ follows by \cite{candes2014solving,hand2017phaselift}, we omit its proof here. We now need to verify $\Omega >0$.
By setting $\delta=\frac{1}{4}$ and $\lambda=(1-s^{*})(\beta_{0}-\alpha_{0})-\frac{5}{4}$, we turn to clarify $C(s)-\Delta>0$ with $C(s) = \mathcal{H}^{*}(s+\gamma)-\frac{l_{0}}{2}\gamma$.

	By Proposition \ref{H}.(b), we have $\mathcal{H}^{*}(s+\gamma)\textgreater \mathcal{H}^{*}(s^{*}-\gamma)$ since $\mathcal{H}^{*}(\cdot)$ is strictly monotonically decreasing and $s\le s^{*}-2\gamma$.
	Now, we just need to explain $\mathcal{H}^{*}(s^{*}-\gamma)-\frac{l_{0}}{2}\gamma-\Delta>0$. When $\mathcal{H}^{*}(s^{*})=0$, by Proposition \ref{H}.(d) we can get 
	\begin{equation*}
		\mathcal{H}^{*}(s^{*}-\gamma)=\mathcal{H}^{*}(s^{*}-\gamma)-\mathcal{H}^{*}(s^{*})\ge l_{0}\gamma.
	\end{equation*}
	Set $\Delta=\frac{l_{0}\gamma}{4}$, we can get $C(s)-\Delta>\frac{l_{0}\gamma}{4}>0$. For sufficiently small $\gamma$ and $l_0$, we have $\Delta=\frac{l_{0}\gamma}{4}\le\frac{s^{*}\gamma}{8}\textless \lambda$, which leads to the number of measurements $m\ge C[\gamma^{-2}\ln\gamma^{-1}]n$ and the success probability at least $ 1-\mathcal{O}(e^{-cm\gamma^{2}})$ when we substitute the parameters into Theorem \ref{l_17} and Lemma \ref{dual1}. The constant $C_{0}$ can also be obtained as 
	\begin{equation*}
		C_{0}=\frac{(18-2\beta_{0}+2\lambda_{0})(1+ \Delta/\lambda_{0}  )}{\Omega}+\frac{18-2\beta_{0}}{\lambda_{0}}\le\frac{80}{l_{0}\gamma}+12.
	\end{equation*}
	
\end{proof}

\section{Counterexample}\label{prf_thm3}
In this section, we will give the proof of Theorem \ref{theorem3} to show that $\ell_1$-minimization  almost fails to exactly recover the matrix $\pmb{X}_0=\pmb{x}_0\pmb{x}_0^{\tp}$ when the fraction $s>s^*\approx 0.1185$.
In Theorem \ref{theorem3}, we remove the distributional assumption compared to Theorem \ref{theorem} and \ref{theorem2} and construct a counterexample. Meanwhile, the numerical experiment confirms our results. The theoretical analysis of sparse outliers without  any distributional assumptions in robust phase retrieval will be studied in the future.
\subsection{Proof of Theorem \ref{theorem3}}
Our proof relies on the following lemma, which can be casted as a variation of  Lemma \ref{X} and we omit its proof here. The main difference is the balance function changes from $H(\rho, s+\gamma)$ to $H(\rho, s-\gamma)$.
\begin{lemma}\label{X2}
	Let $\pmb{X} = \pmb{x} \pmb{y}^{\tp} + \pmb{y} \pmb{x}^{\tp}$ and $\rho=\langle\pmb{x}/\|\pmb{x}\|_{2}, \pmb{y}/\|\pmb{y}\|_{2}\rangle$. For $0 \textless\gamma \textless s$ and $0 \textless s \textless 1$, there exists positive numerical $c(r)$ such that with probability at least $1-\mathcal{O}(e^{-c(r)m\gamma^{2}})$,	
	\begin{eqnarray*}
		\frac{1}{m}\min_{\mS \subset[m],  |\mS|\leq s m}\frac{\Vert \mathcal{A}_{\mS^{c}}(\pmb{X})\Vert_1 -\Vert \mathcal{A}_{\mS}(\pmb{X})\Vert_1}{\Vert \pmb{X}\Vert_F} &\leq& \sqrt{\frac{2}{1+\rho^2}}H(\rho, s-\gamma)+\frac{r\gamma}{2}.
	\end{eqnarray*}
\end{lemma}
\begin{proof}[Proof of Theorem \ref{theorem3}]
	Here, we give a counterexample for any $\pmb{x}\in\mathbb{R}^{n}$. Let $\rho^{*}$ be the corresponding $\rho$ to make $H(\rho,s^{*})$ being zero, that is, $H(\rho^{*},s^{*})=\mathcal{H}(s^{*})=0$. For $H(\rho^{*},s^{*})$, there exists $r_{0}\textgreater0$, such that \begin{eqnarray}\label{Hy}
		H(\rho^{*},s^{*})-H(\rho^{*},s^{*}+\gamma)\ge r_{0}\gamma
	\end{eqnarray}  
holds for any $0 < \gamma\leq 1-s^{*}$.
	Considering the invariance of rotation, we can take $\pmb{x}_0=\alpha \pmb{e}_{1}$. Let $\pmb{x}_{1}=\pmb{e}_{1}$ and $\pmb{y}_{1}=\rho^{*}\pmb{e}_{1} + \sqrt{1-(\rho^{*})^{2}}\pmb{e}_{2}$, thus $\langle\pmb{x}_{1}, \pmb{y}_{1}\rangle=\rho^{*}$.
	
	Now, set $$\pmb{X}_0=\pmb{x}_0\pmb{x}_0^{\tp}=\alpha^{2}\pmb{x}_{1}\pmb{x}_{1}^{\tp}$$
	to be the ground-truth and let 
	\begin{eqnarray*}
		\widehat{\pmb{X}}=\pmb{X}_0+\pmb{H}
		=\pmb{X}_0+\pmb{H}_{T}+\pmb{H}_{T^{\perp}},
	\end{eqnarray*}
	where $\pmb{H}_{T}=
	\begin{bmatrix}
		\widehat{\pmb{H}} &\pmb{0} \\
		\pmb{0} &\pmb{0} 
	\end{bmatrix} \in \mathbb{R}^{n\times n}$, 
	$\widehat{\pmb{H}} =
	\begin{bmatrix}
		2\rho^{*}& \sqrt{1-(\rho^{*})^{2}}\\
		\sqrt{1-(\rho^{*})^{2}}& 0
	\end{bmatrix}\in \mathbb{R}^{2\times 2}$
	and $\pmb{H}_{{T}^{\perp}}=\frac{1-(\rho^{*})^{2}}{\alpha^{2}+2\rho^{*}}\pmb{e}_{2}\pmb{e}_{2}^{\tp}$. Thus, we can get $\widehat{\pmb{X}}\succeq\pmb{0}$.
	
	Let $\mS$ be the support of the largest $sm$ absolute value of entries in $\mathcal{A}(\pmb{H}_{T})$.
	If $\pmb{b} = \mathcal{A}(\pmb{X}_0) + \pmb{z}$ and there is no dense noise $\pmb{\omega}$, we define $\pmb{z} =  \mathcal{A}_\mS(\pmb{H}_{T})$ to be the adversarial sparse outlier and we let the loss function be
	\begin{eqnarray*}
		f(\pmb{X}) = \Vert \mathcal{A}(\pmb{X}) - \pmb{b}\Vert_1
		= \Vert \mathcal{A}(\pmb{X}-\pmb{X}_0) - \mathcal{A}_{\mS}(\pmb{H}_{T})\Vert_1.
	\end{eqnarray*} 
	When $\pmb{X} =\widehat{\pmb{X}}$, the loss is 
	\begin{eqnarray*}
		f(\widehat{\pmb{X}}) 
		&=& \Vert \mathcal{A}(\pmb{H})-\mathcal{A}_{\mS}(\pmb{H}_{T})\Vert_1
		=\Vert\mathcal{A}_{\mS^{c}}(\pmb{H}_{T})+\mathcal{A}(\pmb{H}_{T^{\perp}})\Vert_1\\
		&\le&\Vert\mathcal{A}_{\mS^{c}}(\pmb{H}_{T})\Vert_1+\Vert\mathcal{A}(\pmb{H}_{T^{\perp}})\Vert_1,
	\end{eqnarray*}
	while $f(\pmb{X}_0) = \Vert \mathcal{A}_\mS(\pmb{H}_{T})\Vert_1$ at $\pmb{X} =\pmb{X}_0$.
	
	By Lemma \ref{X2}, for $s\textgreater s^{*}+2\gamma$ we can obtain that 
	\begin{eqnarray}\label{HT}
		\frac{1}{m}\frac{\Vert \mathcal{A}_{\mS^{c}}(\pmb{H}_{T})\Vert_1-\Vert \mathcal{A}_{\mS}(\pmb{H}_{T})\Vert_1 }{\Vert \pmb{H}_{T}\Vert_F} &\le& \sqrt{\frac{2}{1+(\rho^*)^2}}H(\rho^*, s-\gamma)+\frac{r_{0}\gamma}{2}
	\end{eqnarray}
	with probability at least $1-\mathcal{O}(e^{-cm\gamma^{2}})$. To control $\Vert\mathcal{A}(\pmb{H}_{T^{\perp}})\Vert_1$, we use a weakened version of Lemma \ref{RIP_1} that 
 for a fixed PSD $\pmb{X}$, we have
		\begin{eqnarray}\label{rip222}	
				\frac{1}{m} \|\mathcal{A}(\pmb{X})\|_1 \leq  (1+\delta ) \|\pmb{X}\|_1
			\end{eqnarray}	
		with probability at least $1-2e^{-cm\delta^{2}}$ for $0\textless\delta\textless 1$. Here the condition $m=\mathcal{O}(n)$ is not required, and the proof is the application of Bernstein Inequality.
	Thus, combining (\ref{HT}) and (\ref{rip222}), we have	
	\begin{eqnarray*}
		f(\widehat{\pmb{X}})-f(\pmb{X}_0)&\le&\Vert\mathcal{A}_{\mS^{c}}(\pmb{H}_{T})\Vert_1+\Vert\mathcal{A}(\pmb{H}_{T^{\perp}})\Vert_1- \Vert \mathcal{A}_\mS(\pmb{H}_{T})\Vert_1\\
		&\le&m\left[\sqrt{\frac{2}{1+(\rho^*)^2}}H(\rho^*, s-\gamma)+\frac{r_{0}\gamma}{2}\right]\cdot\Vert\pmb{H}_{T}\Vert_{F}+2	\Vert\pmb{H}_{T^{\perp}}\Vert_1\\
		&=&m\left[2 H(\rho^{*}, s-\gamma)+\sqrt{\frac{1+(\rho^*)^2}{2}}r_{0}\gamma\right]+2\frac{1-(\rho^{*})^{2}}{\alpha^{2}+2\rho^{*}}\\
		&\le&m\Big[2H(\rho^{*}, s^{*}+\gamma)+r_{0}\gamma\Big]+\frac{2}{\alpha^{2}}\\
		&\le&-mr_{0}\gamma+\frac{2}{\alpha^{2}}\textless0,
	\end{eqnarray*}
	where in the fourth line we substitute $s\textgreater s^{*}+2\gamma$ and we use (\ref{Hy}) and choose $\alpha=\frac{2}{\sqrt{mr_{0}\gamma}}$ in the last step.
	
	Finally, by scaling the matrix $\widehat{\pmb{X}}$ and $\pmb{H}$, we can construct counterexamples for any $\pmb{x}_0=\alpha \pmb{e}_{1},\alpha\textgreater0$ and then for any $\pmb{x}_0\in\mathbb{R}^{n}$ by rotational invariance.
\end{proof}

\subsection{Experiments}\label{sec_6_2}
At the end of this section, we conduct an experiment to study the fraction of sparse outlier that $\ell_1$-regression can tolerate even when the sample size is very large. 
Here, we let the ground-truth $\pmb{x}_0 = [0.1, 0, 0, 0, 0 ]^{\tp} \in \mathbb{R}^{5}$, the bounded noise $\pmb{\omega} = \pmb{0}$ and the adversarial sparse outlier $\pmb{z} = \mathcal{A}_{\mS}(\pmb{H}_T)$. 
It is selected following from the counterexample in the proof of Theorem \ref{theorem3}.
The index set $\mS$ is the support of the largest $sm$ absolute value of entries in $\mathcal{A}(\pmb{H}_{T})$ and $\pmb{H}_{T}=
\begin{bmatrix}
	\widehat{\pmb{H}} &\pmb{0} \\
	\pmb{0} &\pmb{0} 
\end{bmatrix} \in \mathbb{R}^{5\times 5}$, 
$\widehat{\pmb{H}} =
\begin{bmatrix}
	2\rho^{*}& \sqrt{1-(\rho^{*})^{2}}\\
	\sqrt{1-(\rho^{*})^{2}}& 0
\end{bmatrix}\in \mathbb{R}^{2\times 2}$ with $\rho^* = 0.795$.
We choose the sample size $m = 300n$ and the entries of the sample vectors $\{\pmb{a}_i\}_i$ are i.i.d. drawn from $\mathcal{N}(0, 1)$. In Figure \ref{Figpr11}, we plot the error of the recovered $\pmb{X}$ from $\pmb{X}_0=\pmb{x}_0\pmb{x}_0^{\tp}$ with the fraction of sparse noise ranging from $0$ to $1$. We observe that it can tolerate no more than an $s^*\approx 0.1185$ fraction of corruptions which meets the results in the preceding article.
\begin{figure}[h]
	\centering
	\includegraphics[scale=0.35]{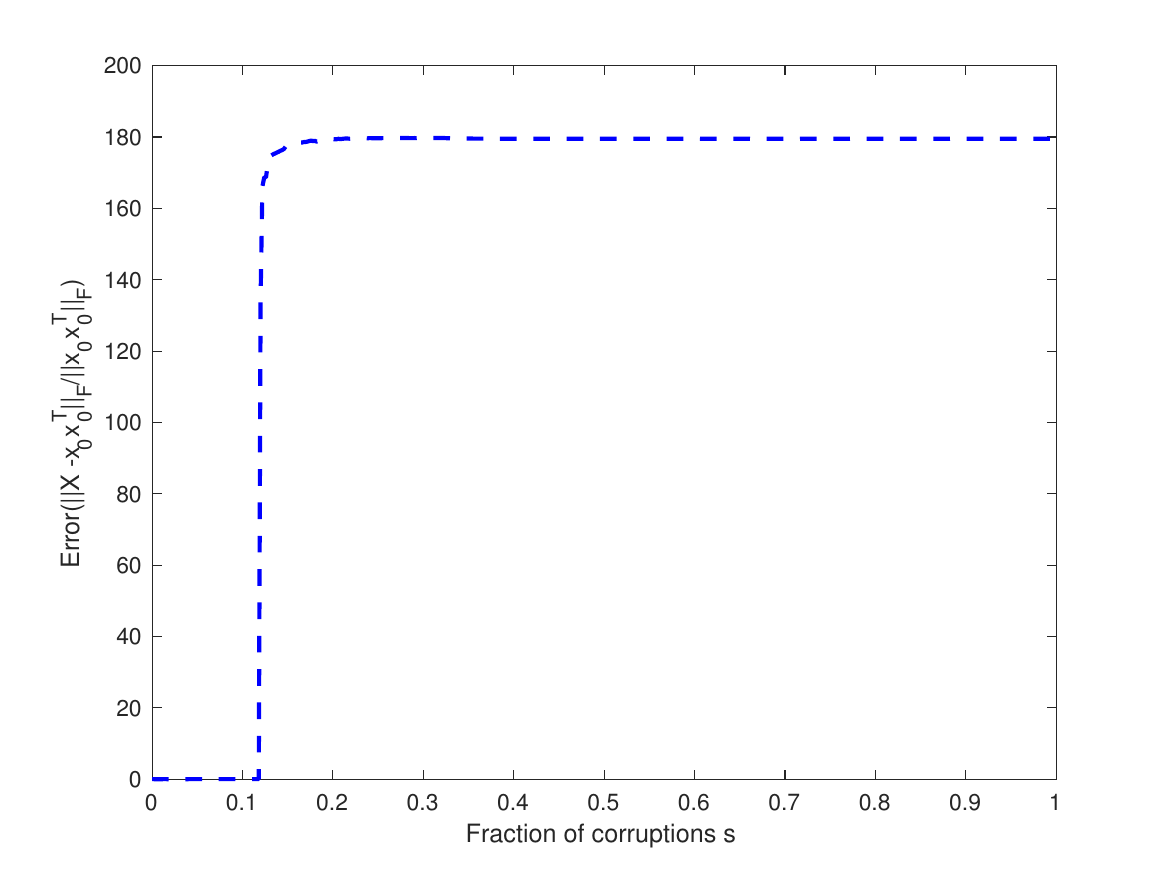}
	\caption{The recovery threshold for $\ell_1$-regression, empirically in n = 5 and  $\pmb{X}_0 = 0.01\pmb{e}_1\pmb{e}_1^{\tp}$.}
	\label{Figpr11}
\end{figure}

\appendix
\section{Proof for Auxiliary Lemmas and Propositions}\label{appen_prfs}

\subsection{Auxiliary Lemmas}
Firstly, we will exhibit some auxiliary lemmas that appear in the proof.
We first display a  corollary of standard Dvoretzky-Kiefer-Wolfowitz inequality \cite{massart1990tight}, similar proofs can be found in \cite{karmalkar2019compressed,xu2022low}.
\begin{lemma}\label{l_21}
	Let $\Gamma=\{ \xi_1, \dots, \xi_m\}$ be i.i.d. sampled from any distribution function $F$ on $\mathbb{R}$. Then for any $\eta, \epsilon \in [0, 1]$, the following holds with probability at least $1 - 4e^{-2m\epsilon^2}$,
	\begin{eqnarray*}
		F^{-1}(\eta- \epsilon) < \widehat{F}^{-1}(\eta) < F^{-1}(\eta+ \epsilon).
	\end{eqnarray*}
	Here, $\widehat{F}$ is the associated empirical distribution function defined by $\widehat{F}(x) = {1\over m} \sum_{i=1}^m\mathbbm{1}_{\xi_i \leq x}$.
\end{lemma} 
In addition, Bernstein Inequality will play a significant role in this article, below we present the generalized form that do not require random variables to be mean-zero.
\begin{lemma}\cite{vershynin2018high}\label{l_20}
	Let $X_1,\ldots, X_N$ be independent sub-exponential random variables. Let $K:=\max_{i}\|X_{i}- \mathbb{E}X_i\|_{\psi_1}$. Then, for every $ t\ge0$, we have
	\begin{eqnarray*}
		\mathbb{P}\left(\left\vert \frac{1}{N} \sum_{i = 1}^{N}( X_i -  \mathbb{E}X_i)\right\vert
		\ge t\right) \leq 
		2 e^{-c\, N\min\Bigl(\frac{t^2}{K^{2}},
			\frac{t}{K}\Bigr)},
	\end{eqnarray*}
	where $c\textgreater 0$ is a positive  constant. 
\end{lemma}
The following lemma is the classical conclusion in non-asymptotic random matrices.
\begin{lemma}\cite{vershynin2010introduction}\label{random}
	Let $\pmb{A}$ be an $N\times n$ matrix whose rows $\pmb{A}_{i}$ are independent sub-Gaussian isotropic random vectors in $\mathbb{R}^{n}$. Then for every $ t\ge 0$, with probability at least $1-2e^{-ct^{2}}$, we have
	\begin{equation}
		\left\|\frac{1}{N}\pmb{A}^{\tp}\pmb{A}-\pmb{I}\right\|\le \max(\delta,\delta^{2}),\quad \text{where}\quad \delta=C\sqrt{\frac{n}{N}}+\frac{t}{\sqrt{N}}.
	\end{equation}
	Here $C=C_{K},c=c_{K}>0$ only depend on the sub-Gaussian norm $K=\max_{i}\|\pmb{A}_{i}\|_{\psi_{2}}$ of the rows. 
	When $\pmb{A}$ is an $N\times n$ matrix whose entries are independent standard normal random variables, the above formula is established with $C=1$ and with probability at least $1-2e^{-t^{2}/2}$.
\end{lemma}

The following lemma originates from \cite{candes2013phaselift}, and is a strengthened result that holding for all symmetric matrices. We give this outcome directly.
\begin{lemma}\label{RIP_1}
	Fix any $\delta\in (0,\frac{1}{2})$ and assume $m\ge 20\delta^{-2}n$. Then for all symmetric matrices $\pmb{X}$,
	\begin{eqnarray}\label{rip}	
		\frac{1}{m} \|\mathcal{A}(\pmb{X})\|_1 \leq  (1+\delta ) \|\pmb{X}\|_1
	\end{eqnarray}	
	on an event $E_{\delta}$ of probability at least $1-2e^{-m\delta/8}$.
\end{lemma}

\subsection{Proof of Proposition \ref{H}}\label{A_2}

\begin{proof}[Proof of Proposition \ref{H}]
	We first prove that $$\mathbb{E}[|Z_{\rho}|]=\frac{2}{\pi}\Bigl[\sqrt{1-\rho^{2}}+\rho\arcsin(\rho)\Bigr].$$
	Set $\rho=\cos\theta$ in which $\theta\in[0,\pi/2] $. Let $W\sim\mathcal{N}(0,1)$ be  a random variable independent of $X$. Then $Y$ can be expressed as $Y=X\cos\theta+W\sin\theta$. Thus 
	\begin{eqnarray*}
		|Z_{\rho}|=|XY|=|X^{2}\cos\theta+XW\sin\theta|.	
	\end{eqnarray*}	
	Since $X$ and $W$ are independent, we have 
	\bea
	\mathbb{E}[|Z_{\rho}|] = \frac{1}{2\pi} \int_{-\infty}^{+\infty}\int_{-\infty}^{+\infty}\left| X^{2}\cos\theta+XW\sin\theta \right|\cdot e^{-X^2/2}\cdot e^{-W^2/2}dX dW.
	\eea
	Polar coordinates 
	$
	\begin{cases}
		X=r\cos\phi \\
		W=r\sin\phi
	\end{cases} 	
	$
	denote that
	\bea
	\mathbb{E}[|Z_{\rho}|]&=&\frac{1}{2\pi}\int_{0}^{\infty}r^{3}e^{-r^{2}/2}dr\int_{0}^{2\pi}|\cos\phi^{2}\cos\theta+\cos\phi\sin\phi\sin\theta|d\phi\\
	&=&\frac{1}{2\pi}\int_{0}^{2\pi}|\cos\theta+\cos(2\phi-\theta)|d\phi.
	\eea
	Here, we use the fact that $\int_{0}^{\infty}r^{3}e^{-r^{2}/2}dr=2$ and $2(\cos\phi^{2}\cos\theta+\cos\phi\sin\phi\sin\theta)=2\cos\phi\cos(\theta-\phi)=\cos\theta+\cos(2\phi-\theta)$.
	
	Then, based on the properties of symmetry, periodicity and $\tilde{\phi} = 2\phi -\theta$, we have
	\bea
	\mathbb{E}[|Z_{\rho}|] &=& \frac{1}{4\pi}\int_{\theta}^{4\pi-\theta}|\rho+\cos\tilde{\phi}|d\tilde{\phi} = \frac{1}{\pi}\int_{0}^{\pi}|\rho+\cos\tilde{\phi}|d\tilde{\phi} \\
	&=&\frac{1}{\pi}\int_{0}^{\pi}|\rho-\cos\tilde{\phi}|d\tilde{\phi}
	=\frac{1}{\pi}\left[\int_{0}^{\theta}(\cos\tilde{\phi}-\rho)d\tilde{\phi}+\int_{\theta}^{\pi}(\rho-\cos\tilde{\phi})d\tilde{\phi}\right]\\
	&=&\frac{2}{\pi}\Bigl[\sin\theta+\rho(\pi/2-\theta)\Bigr]
	=\frac{2}{\pi}\Bigl[\sqrt{1-\rho^{2}}+\rho\arcsin(\rho)\Bigr].
	\eea
	
	\noindent(a)
	By definition, $\mathcal{H}^{*}(0)=\min\limits_{\rho \in[0,1]}\sqrt{\frac{2}{1+\rho^{2}}}\mathbb{E}[|Z_{\rho}|]$ and $\mathcal{H}^{*}(1)=-\max\limits_{\rho \in[0,1]}\sqrt{\frac{2}{1+\rho^{2}}}\mathbb{E}[|Z_{\rho}|]$.\\
	Let
	\begin{eqnarray*} g(\rho)=\sqrt{\frac{2}{1+\rho^{2}}}\mathbb{E}[|Z_{\rho}|]&=&\frac{2\sqrt{2}}{\pi}\left[\sqrt{\frac{1-\rho^{2}}{1+\rho^{2}}}+\frac{\rho\arcsin(\rho)}{\sqrt{1+\rho^{2}}}\right].\\
	\end{eqnarray*}	
	Deriving $g(\rho)$,
	\bea
	g'(\rho) = \frac{2\sqrt{2}}{\pi} \frac{ -\rho \sqrt{1- \rho^{2}}+ \arcsin(\rho) }{(1+\rho^{2})^{3/2}} \overset{(i)}{\geq} \frac{2\sqrt{2}}{\pi} \frac{ \rho(1- \sqrt{1- \rho^{2}})}{(1+\rho^{2})^{3/2}}\geq 0,
	\eea
	where the inequality (i) holds since $\arcsin(\rho)\ge \rho$ when $\rho \in [0, 1]$.
	Thus, we can get 
	\begin{eqnarray*} 
		\mathcal{H}^{*}(0)=g(0)=\frac{2\sqrt{2}}{\pi}\quad\text{and}\quad\mathcal{H}^{*}(1)=-g(1)=-1.
	\end{eqnarray*}
	\\
	\noindent(b) 
	The continuity is obvious.
	Let $0\le s_{1} \textless s_{2} \le1$ and $\rho_{1} = \arg\min \mathcal{H}^{*}(s_{1})$, $\rho_{2} = \arg\min \mathcal{H}^{*}(s_{2})$.
	By monotonicity and definition, 
	\begin{eqnarray*} 
		\mathcal{H}^{*}(s_{1})=\sqrt{\frac{2}{1+\rho_1^2}}H(\rho_{1},s_{1})\textgreater\sqrt{\frac{2}{1+\rho_1^2}}H(\rho_{1},s_{2})\ge\sqrt{\frac{2}{1+\rho_2^2}}H(\rho_{2},s_{2})=\mathcal{H}^{*}(s_{2}).
	\end{eqnarray*}
	Thus, $\mathcal{H}^{*}(s)$ is strictly monotonically decreasing. By Proposition \ref{H}.(a), $\mathcal{H}^{*}(0)$ and $ \mathcal{H}^{*}(1)$ are finite, which leads to $\mathcal{H}^{*}(s)$ being well defined. \\ 
	\noindent(c) 
	This conclusion can be obtained from Proposition \ref{H}.(b) that $\mathcal{H}^{*}(s)$ is a monotone uniformly continuous function.\\
	\noindent(d) 
	The uniqueness is obvious by Proposition \ref{H}.(b). We calculate the balance ratio $s_{\rho}$ from $H(\rho, s_{\rho}) =0$ for any $0 \leq \rho \leq 1$, that is, $\int_0^{F_{\rho}^{-1}(1-s)}zf_{\rho}(z)dz = \int_{F_{\rho}^{-1}(1-s)}^{+\infty}zf_{\rho}(z)dz$ holds for any $0 \leq \rho \leq 1$. And we denote the minimal balance ratio
	\bea
	s^* = \min_{0 \leq \rho \leq 1} s_{\rho}.
	\eea

	\begin{figure}[H]
		\centering 
		\begin{minipage}[b]{0.4\textwidth} 
			\centering 
			\includegraphics[width=1\textwidth]{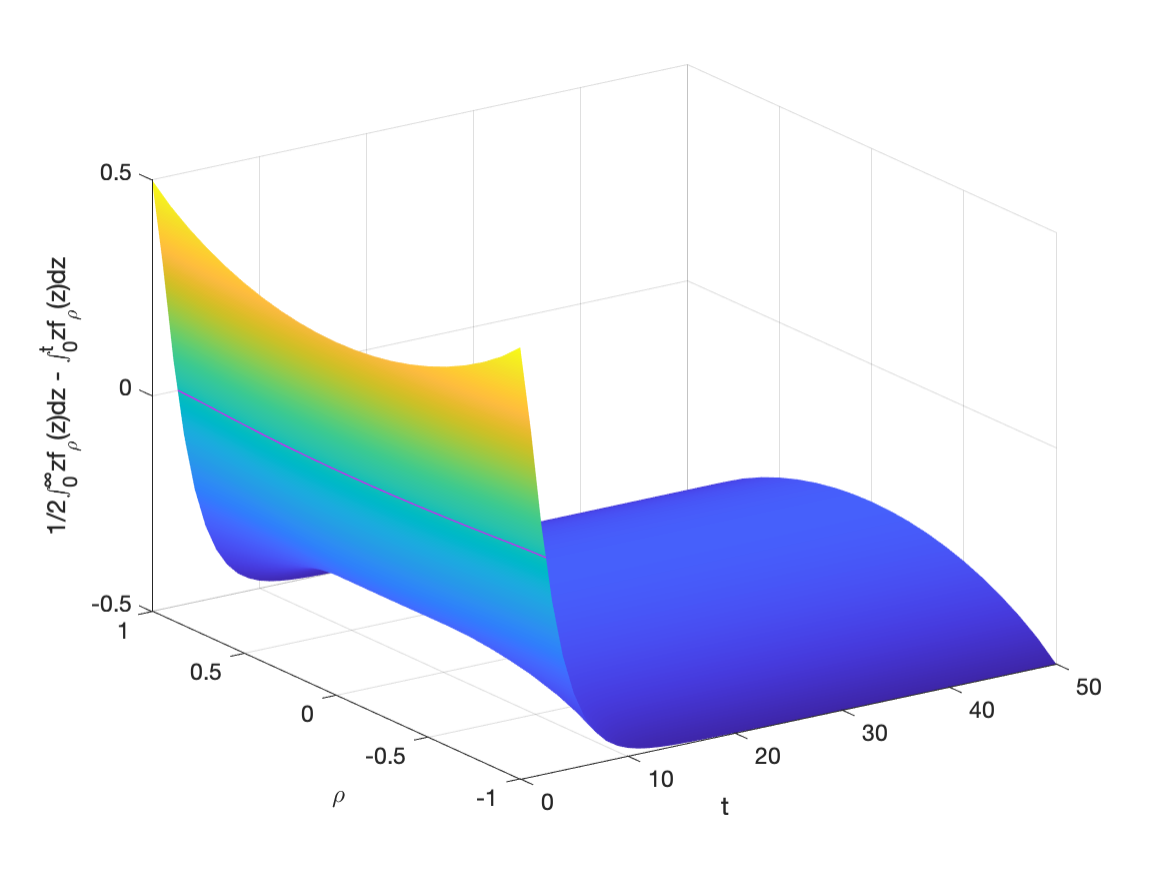}
			\caption{The value of $\frac{1}{2}\int_{0}^{\infty} zf_{\rho}(z)dz - \int_{0}^{t} z f_{\rho}(z)dz$.}
			\label{Figpr6}
		\end{minipage}
		\begin{minipage}[b]{0.4\textwidth} 
			\centering 
			\includegraphics[width=1\textwidth]{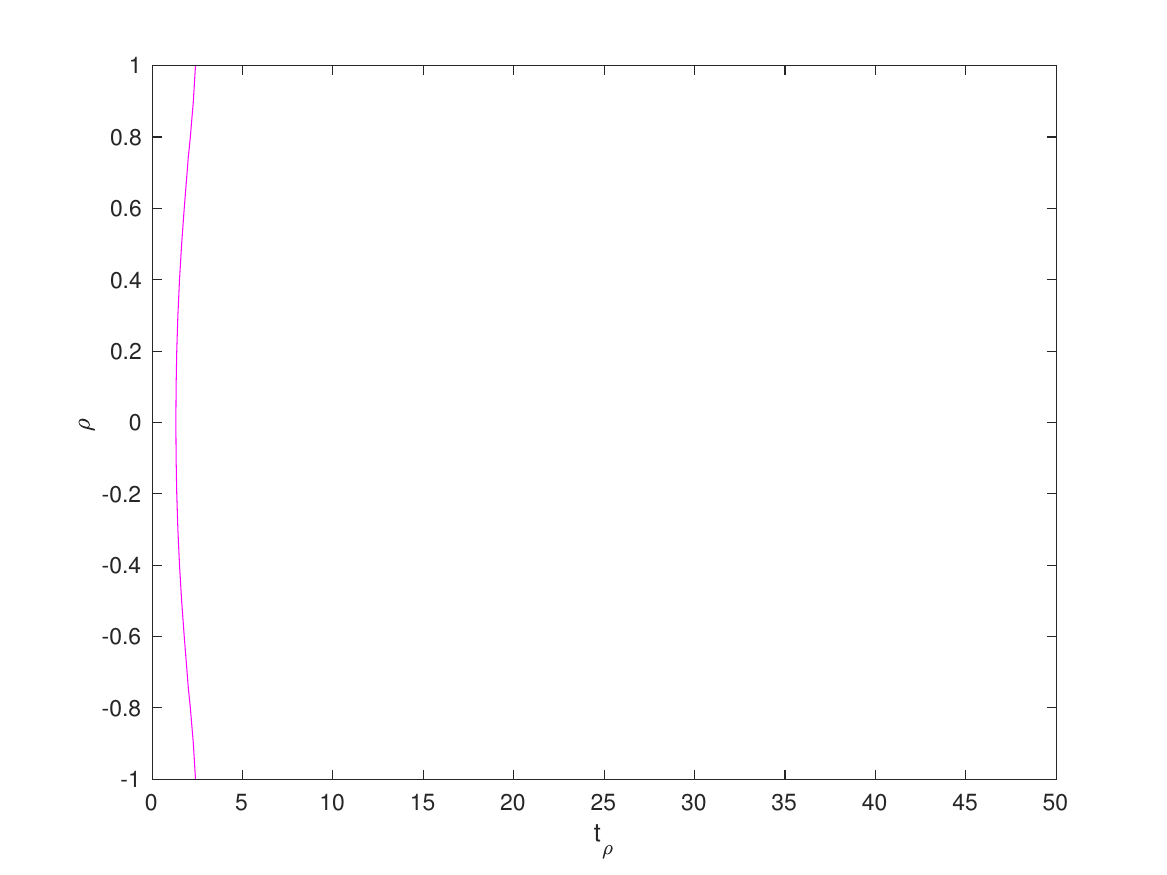}
			\caption{The curve $G(\rho, t_{\rho}) = 0$ in Figure \ref{Figpr6}.}
		\label{Figpr7}
	\end{minipage}
	
\end{figure}

	However, it is extremely complicated to directly compute $s_{\rho}$ and $s^*$. We split this problem and illustrate it in the form of images. Firstly, we try to get $t_{\rho}$ to make the integral $\int_{0}^{\infty} z f_{\rho}(z) dz$ divided into two balanced parts, that is, $\int_{0}^{t_{\rho}} z f_{\rho}(z) dz = \int_{t_{\rho}}^{\infty} z f_{\rho}(z) dz$. In Figure \ref{Figpr6}, we plot the graph of $\frac{1}{2}\int_{0}^{\infty} zf_{\rho}(z)dz - \int_{0}^{t} zf_{\rho}(z)dz$. Setting the formula to be $0$, we can get a curve $G(\rho, t_{\rho}) = 0$ in Figure \ref{Figpr7}. 
Since $s_{\rho} =1- F_{\rho}(t_{\rho}) = 1 - \int_{0}^{t_{\rho}} f_{\rho}(z) dz$, we can calculate $(1 - s_{\rho})$
from the intersection line of the cylindrical surface $G(\rho, t_{\rho}) = 0$ and the CDF surface $\int_{0}^tf_{\rho}(z)dz$ (see more details in Figure \ref{Figpr10}), then obtain the minimal balance ratio
$$s^* = \min_{0 \leq \rho \leq 1} s_{\rho} \approx 0.1185.$$
Essentially, it is the minimum of $s_{\rho} =1- F_{\rho}(t_{\rho})$ in $G(\rho, t_{\rho}) = 0$. And $\rho^* \approx 0.795$ derives $H(\rho^*,s^*) =0$.

Consider the equation $\mathcal{H}^{*}(s^*-\gamma)-\mathcal{H}^{*}(s^*) \geq l_0 \gamma$. We can easily know that the equation holds for some $l_0\geq 0$ since $\mathcal{H}^{*}(\cdot)$ is monotonously decreasing. If we need the equation holds for some $l_0> 0$, we could see that the derivative of $\mathcal{H}^{*}(\cdot)$ in $s^*$ is nonzero. It is nearly impossible to compute the closed-form of the derivative of $\mathcal{H}^{*}(\cdot)$. Here, we give an analytical method. Based on the definition of $\mathcal{H}^{*}(\cdot)$, we have
\bea
\mathcal{H}^{*}(s) = \min_{\rho\in[0,1]}\sqrt{\frac{2}{1+\rho^2}}H(\rho,s) = \min_{\rho\in[0,1]}\sqrt{\frac{8}{1+\rho^2}}\left[\int_{0}^{F_{\rho}^{-1}(1-s)} zf_{\rho}(z)dz - \frac{1}{2}\int_{0}^{\infty} zf_{\rho}(z)dz \right].
\eea
We then to consider $\frac{\partial H(\rho,s)}{\partial s} = \sqrt{\frac{8}{1+\rho^2}} \frac{\partial \left(\frac{1}{2}\int_{0}^{\infty} zf_{\rho}(z)dz - \int_{0}^{t} zf_{\rho}(z)dz\right)}{\partial t} \frac{\partial F_{\rho}^{-1}(1-s)}{\partial s}$. In Figure \ref{Figpr6}, we plot the graph of $\frac{1}{2}\int_{0}^{\infty} zf_{\rho}(z)dz - \int_{0}^{t} zf_{\rho}(z)dz$ and get that $\frac{\partial \left(\frac{1}{2}\int_{0}^{\infty} zf_{\rho}(z)dz - \int_{0}^{t} zf_{\rho}(z)dz\right)}{\partial t}$ is nonzero on the curve $G(\rho, t_{\rho}) = 0$. Based on the proposition of the CDF $F_{\rho}$, we know that $\frac{\partial F_{\rho}^{-1}(1-s)}{\partial s}$ is nonzero in $s^*$. Thus, we can get that the value of $\sqrt{\frac{8}{1+\rho^2}} \frac{\partial H(\rho,s)}{\partial s}$ is nonzero. Since $\text{argmin}\sqrt{\frac{2}{1+\rho^2}}H(\rho,s) \in [0,1]$, we can obtain that
the derivative of $\mathcal{H}^{*}(\cdot)$ in $s^*$ is nonzero, which implies $\mathcal{H}^{*}(s^*-\gamma)-\mathcal{H}^{*}(s^*) \geq l_0 \gamma$ holds for some $l_0> 0$.

\begin{figure}[h] 
	\centering
	\includegraphics[scale=0.5]{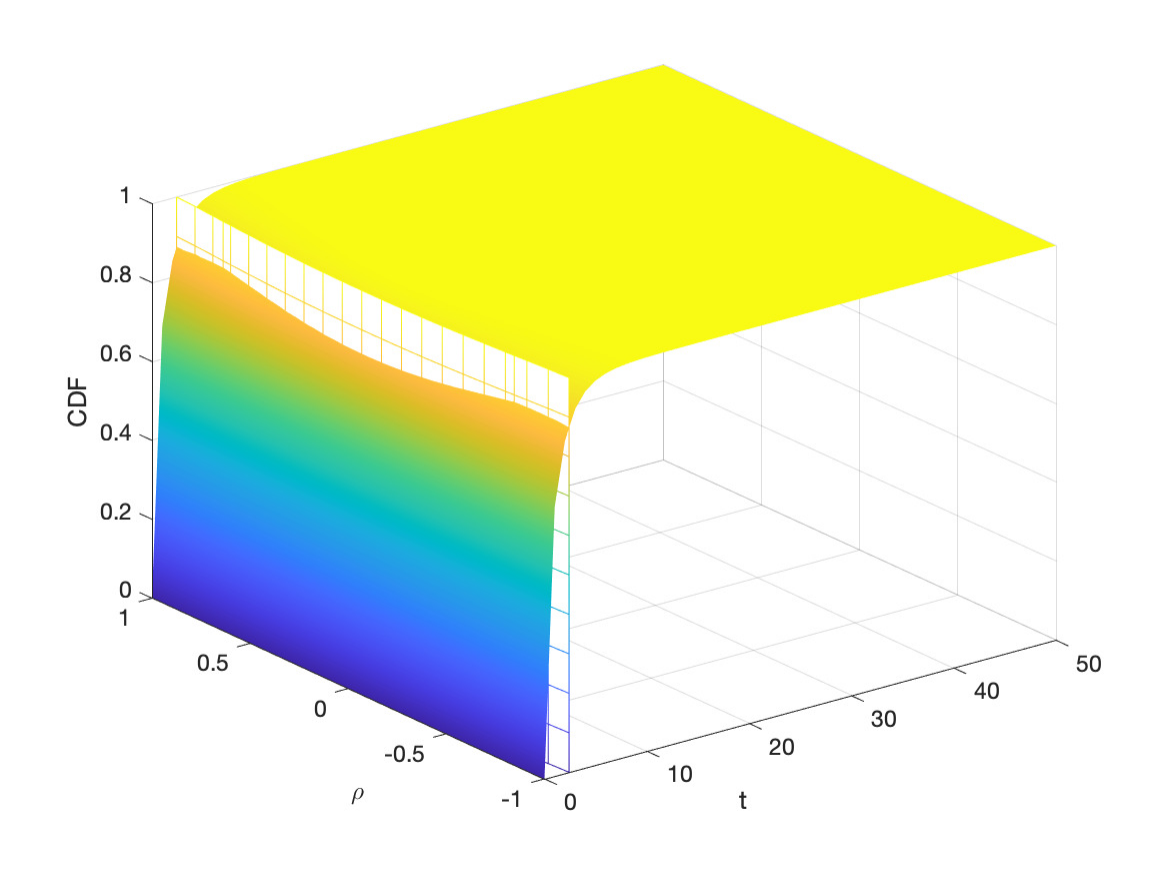}
	\caption{The intersection line of these two surfaces.}
	\label{Figpr10}
\end{figure}

\end{proof}

\subsection{Proof of  (\ref{Z})}\label{ZP}
Without loss of generality we may assume that $	\|X\|_{\psi_{2}}=1$(although it has a specific value). Thus, by denifition, it only need to show that $\mathbb{E}e^{\lvert Z_{\rho}\rvert}\leq 2$.
$Z_{\rho }=X\cdot Y$
yields
\begin{eqnarray*}
\mathbb{E} e^{\lvert Z_{\rho}\rvert} \leq \mathbb{E}e^{X^{2}/2+Y^{2}/2}
=	\mathbb{E}\left[e^{X^{2}/2}\cdot e^{Y^{2}/2}\right]
\leq 	\frac{1}{2}\mathbb{E}\left[e^{X^{2}}+e^{Y^{2}}\right]
=	2.
\end{eqnarray*}

\subsection{Proof of Lemma \ref{a_4}}\label{aaa}
\begin{proof}[Proof of Lemma \ref{a_4}]
	Our proof mainly uses the techniques of covering argument to 
	reform Theorem 5.39 in \cite{vershynin2010introduction}.
	We will control $\|\pmb{\Phi} \pmb{x}\|_{2}$ for all vectors $\pmb{x}$  on a $\frac{1}{4}$-net $\mathcal{N}$ of the unit sphere $\mathcal{S}^{n-1}$:
	\begin{eqnarray*}
		\left\lVert\pmb{\Phi}\right\rVert=\sup_{\pmb{x}\in\mathcal{S}^{n-1}}|\langle\pmb{\Phi} \pmb{x},\pmb{x} \rangle|
		\le	2\max_{\pmb{x}\in\mathcal{N}}|\langle\pmb{\Phi} \pmb{x},\pmb{x} \rangle|
		=2\max_{\pmb{x}\in\mathcal{N}} \left|\frac{1}{N}\sum_{i\in [N]}\varepsilon _{i}\langle \pmb{a}_{i},\pmb{x} \rangle^{2}\right|.
	\end{eqnarray*}	
	Thus, to complete the proof with the required probability, it suffices to show that
	\begin{eqnarray*}
		\max_{\pmb{x}\in\mathcal{N}}| \langle\pmb{\Phi} \pmb{x},\pmb{x} \rangle|=
		\max_{\pmb{x}\in\mathcal{N}}\left| \frac{1}{N}\sum_{i\in [N]}\varepsilon _{i}\langle \pmb{a}_{i},\pmb{x} \rangle^{2}\right|\le \frac{t}{2}
	\end{eqnarray*}
	with high probability, where we know from standard covering number theory that the net $\mathcal{N}$ has cardinality $|\mathcal{N}|\le 9^{n}$.	
	
	When we fix $\pmb{x}\in \mathcal{S}^{n-1}$, we can see that $X:=\langle \pmb{a},\pmb{x} \rangle\sim \mathcal{N}(0,1)$ is Gaussian  random variable and $ Z:=\varepsilon X^{2}$ is mean-zero symmetrical sub-exponential random variable with sub-exponential norm:
	\begin{eqnarray*}
		K_{0}:=\|Z\|_{\psi_{1}}=\|X^{2}\|_{\psi_{1}}=\|X\|^{2}_{\psi_{2}}
	 = \frac{8}{3},
	\end{eqnarray*}
	where the definition of the sub-exponential norm is used in the last step.
	Therefore, by Lemma \ref{l_20}, the Bernstein Inequality gives for every $t\ge0$ 
	\begin{eqnarray*}
		\mathbb{P}\left(|\langle\pmb{\Phi} \pmb{x},\pmb{x} \rangle|\ge\frac{t}{2}\right)= \mathbb{P}\left(\left|\frac{1}{N}\sum_{i\in [N]} Z_{i}\right|\ge\frac{t}{2}\right)
		\le 2e^{-c_{1}N\min\left(\frac{t^{2}}{4K_{0}^{2}},\frac{t}{2K_{0}}\right)}.
	\end{eqnarray*}
	Then taking the union bound over all vectors $\pmb{x}$ in the net $\mathcal{N}$, we obtain
	\begin{eqnarray*}
		\mathbb{P}\left(	\max_{\pmb{x}\in\mathcal{N}}\left| \frac{1}{N}\sum_{i\in [N]}\varepsilon _{i}X_{i}^{2}\right|\ge \frac{t}{2}\right)\le 9^{n}\cdot 2e^{-c_{1}N\min\left(\frac{t^{2}}{4K_0^2},\frac{t}{2K_0}\right)} = 2e^{-cN\min\left(\frac{t^{2}}{4K_0^2},\frac{t}{2K_0}\right)},
	\end{eqnarray*}
	provided $N\gtrsim \min\left(\frac{t^{2}}{4K_0^2},\frac{t}{2K_0}\right)^{-1}n$.
\end{proof}

 \section{Numerical Experiments} \label{experiments}
\begin{figure}[htbp]
\centering 

\begin{minipage}[b]{0.4\textwidth}
	\centering 
	\includegraphics[width=1\textwidth]{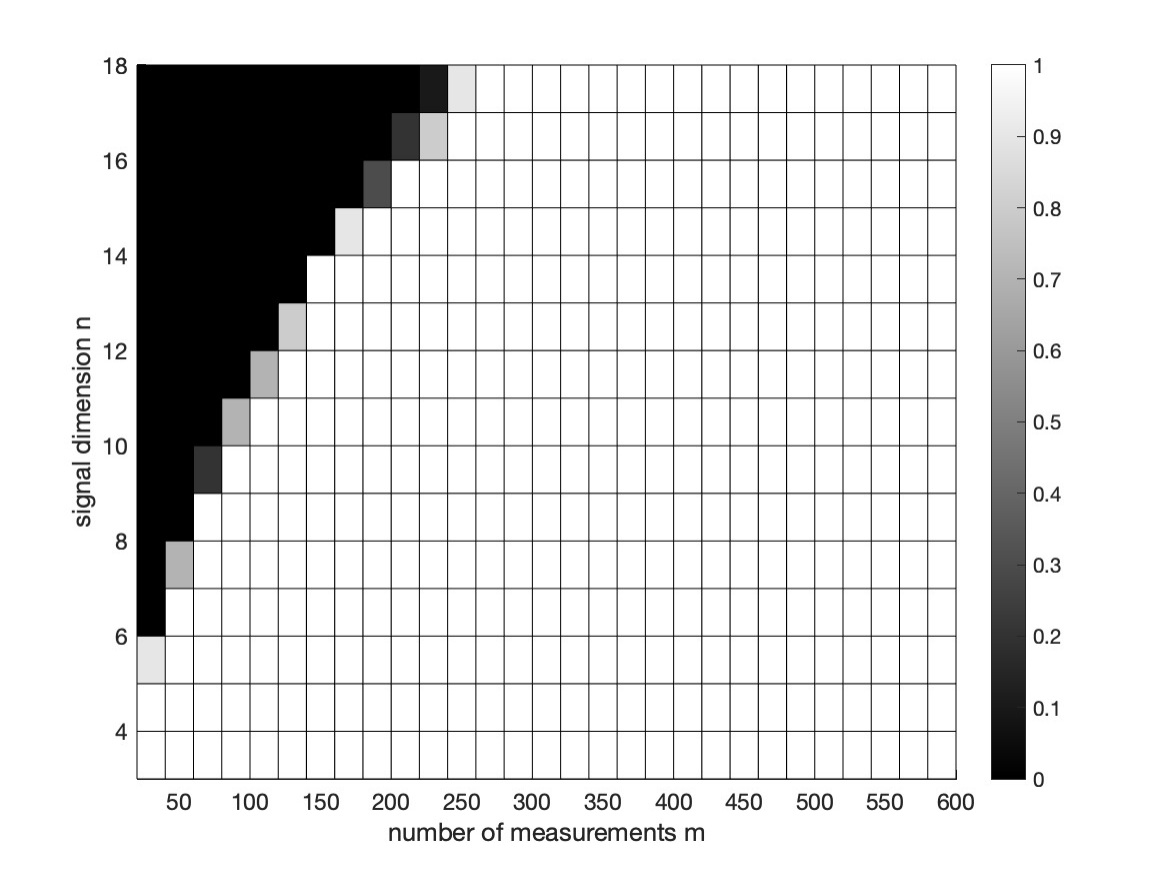} 
	\subcaption*{(a) $s=0$} 
\end{minipage}
\quad
\begin{minipage}[b]{0.4\textwidth} 
	\centering 
	\includegraphics[width=1\textwidth]{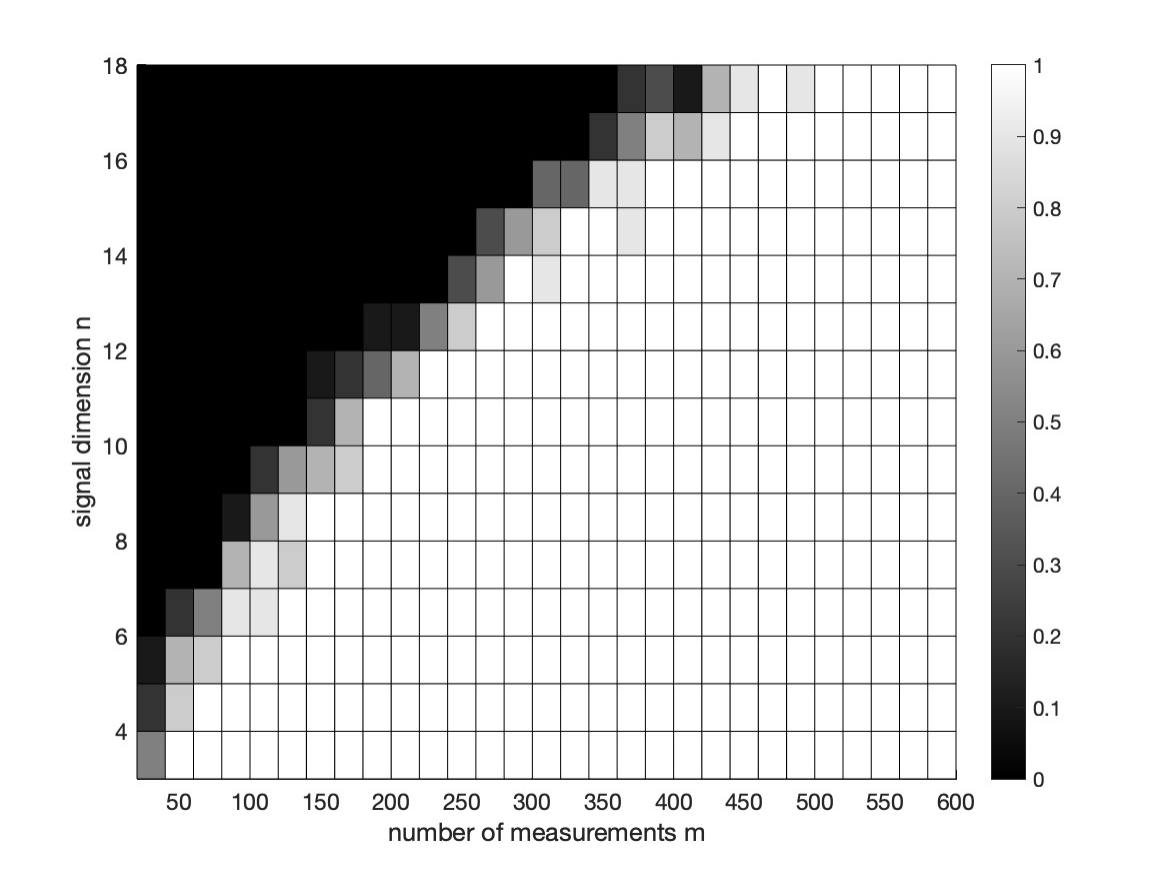}
	\subcaption*{(b) $s=0.05$}
	
\end{minipage}

\begin{minipage}[b]{0.4\textwidth}
	\centering 
	\includegraphics[width=1\textwidth]{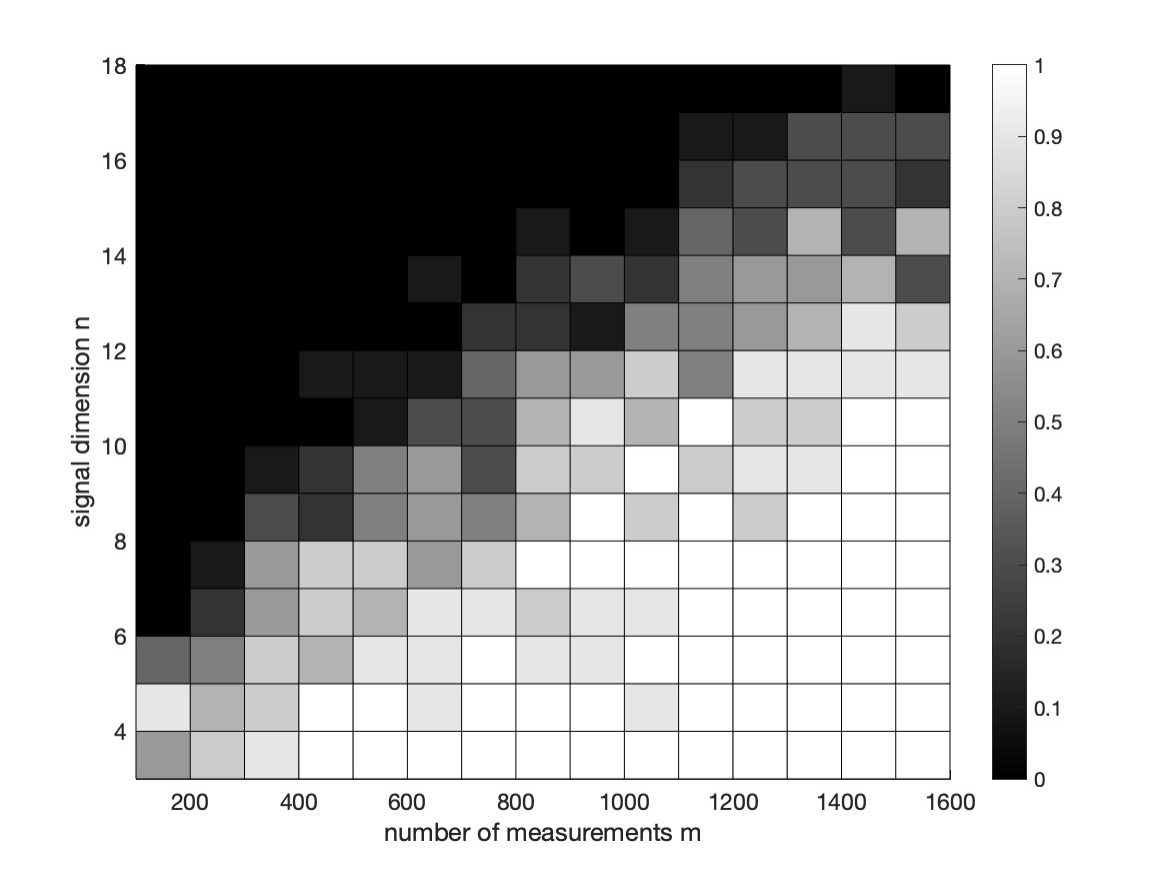} 
	\subcaption*{(c) $s=0.1$} 
\end{minipage}
\quad
\begin{minipage}[b]{0.4\textwidth} 
	\centering 
	\includegraphics[width=1\textwidth]{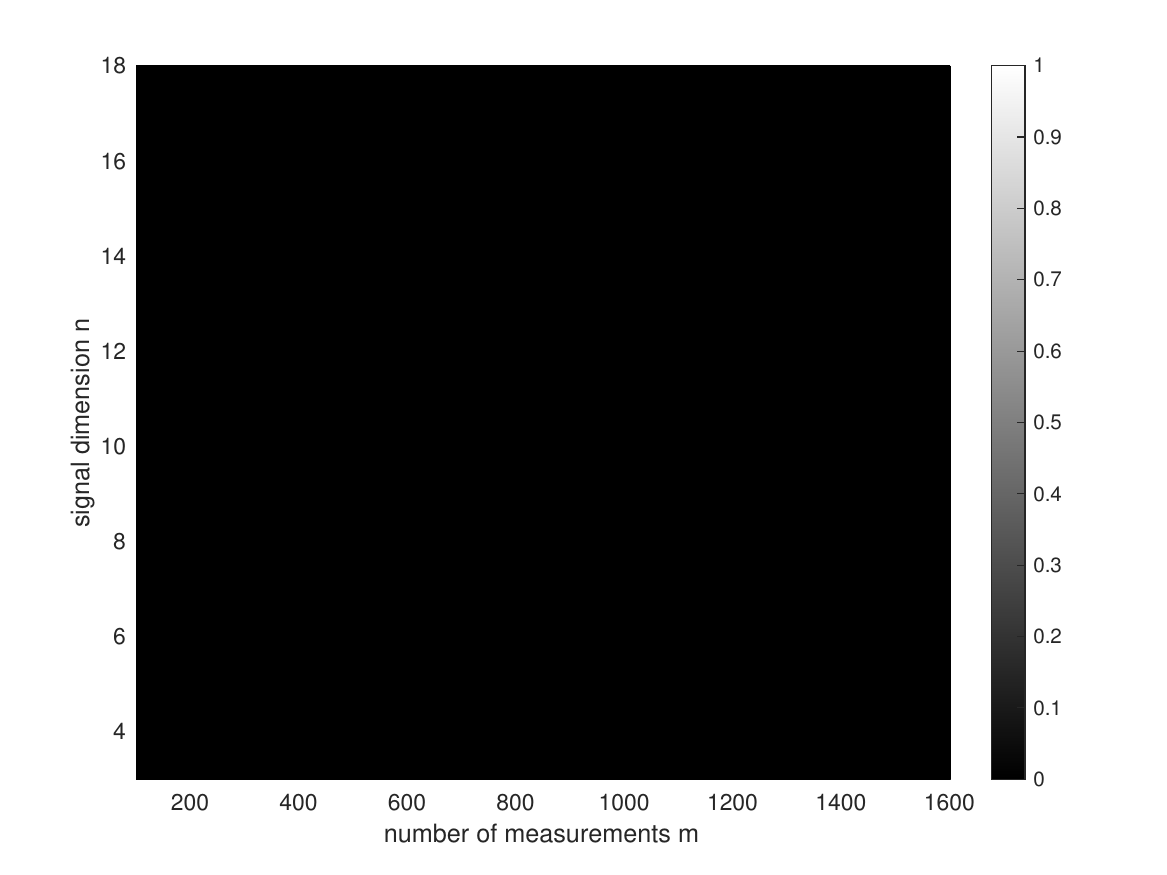}
	\subcaption*{(d) $s=0.15$}
	
\end{minipage}
\caption{Recovery success rate for the Robust-PhaseLift problem (\ref{l1_min}). Black represents a recovery rate of zero while white represents $100\%$. }
\label{Fig_successrate}
\end{figure}
In this section, we use numerical simulation to examine the performance of Robust-PhaseLift \eqref{l1_min}. Let the signal length $n$ vary from $3$ to $17$ and let the number of measurements $m$ vary from $20$ to $600$ (or $100$ to $1700$). Denote the bounded noise $\pmb{\omega} =\pmb{0}$. The sample vectors $\{ \pmb{a}_i \}_i$ are composed of i.i.d. standard Gaussian variables. For each $(n,m)$, we conduct the ground-truth $\pmb{x}_0 = [0.1, 0, \dots, 0] \in \mathbb{R}^n$ and the adversarial sparse outliers are also selected following from the counterexample in the proof of Theorem \ref{theorem3} as $\pmb{z} = \mathcal{A}_{\mS}(\pmb{H}_T)$ where 
$\pmb{H}_{T}=
\begin{bmatrix}
\widehat{\pmb{H}} &\pmb{0} \\
\pmb{0} &\pmb{0} 
\end{bmatrix} \in \mathbb{R}^{n\times n}$, 
$\widehat{\pmb{H}} =
\begin{bmatrix}
2\rho^{*}& \sqrt{1-(\rho^{*})^{2}}\\
\sqrt{1-(\rho^{*})^{2}}& 0
\end{bmatrix}\in \mathbb{R}^{2\times 2}$ and $\rho^* = 0.795$. 
Here, the index set $\mS$ is the support of the largest $sm$ absolute value of entries in $\mathcal{A}(\pmb{H}_T)$.
We attempt to recover $\pmb{X}_0 = \pmb{x}_0\pmb{x}_0^{\tp}$ by solving \eqref{l1_min} using the Quasi-Newton method with different fraction $s$. For each cell, the success rate is calculated by over $10$ independent trials\footnote{If the relative error is less than 0.1, we say that the recovery is successful in a trial.}.

Figure \ref{Fig_successrate} shows that 	
the number of measurements required for successful recovery appears to be linear in $n$. And we can see that the Robust-PhaseLift model can recover the signal when $s = 0, 0.05, 0.1$, which also meets our Theorem \ref{theorem}. Besides, we find that the number of measurements required for successful recovery increases along with the growing of noise fraction $s$. It is also reasonable that we can obtain the number of measurements $m\geq C[\gamma^{-2}\ln\gamma^{-1}]n$ increases when $s$ grows, $\gamma$ decreases respectively.

		\normalem
	\bibliographystyle{plain}
	\bibliography{ref}
	
\end{document}